\documentclass{amsart}
\usepackage[normalem]{ulem}
\usepackage{amsmath}
\usepackage{setspace}
\usepackage{fullpage,graphicx,mathpazo,color}
\usepackage{amsmath,amscd}
\usepackage{mathrsfs}
\usepackage{subfigure}
\usepackage{amsmath}
\usepackage{arydshln}
\usepackage{booktabs}
\usepackage{tikz}
\usepackage{lscape}
\newcommand{\ii}{\mathrm{i}}
\newcommand{\ee}{\mathrm{e}}

\newcommand{\T}{\mathrm{T}}
\newcommand{\sech}{\mathrm{sech}}

\newtheorem{rhp}{Riemann-Hilbert Problem}
\newtheorem{theorem}{Theorem}
\newtheorem{lemma}{Lemma}
\newtheorem{prop}{Proposition}

\newtheorem{remark}{Remark}
\newtheorem{example}{Example}

\usepackage{graphicx}      
\usepackage{titlesec}

\newcommand{\lm}[1]{{\color{black} #1}}
\newcommand{\ys}[1]{{\color{black} #1}}

\titleformat{\title}{\centering\LARGE\bfseries}{\thetitle}{1em}{}
\titleformat{\section}{\centering\LARGE\bfseries}{\thesection}{1em}{}
\titleformat{\subsection}{\Large\bfseries}{\thesubsection}{1em}{}
\begin{document}

\title{Asymptotic analysis of high-order soliton for the Hirota equation}
\author{Xiaoen Zhang}
\address{School of Mathematics, South China University of Technology, Guangzhou, China, 510641}
\author{Liming Ling}
\address{School of Mathematics, South China University of Technology, Guangzhou, China, 510641}
\email{Corresponding author:linglm@scut.edu.cn}

\begin{abstract}
In this paper, we mainly analyze the long-time asymptotics of high-order soliton for the Hirota equation. Two different Riemann-Hilbert representations of Darboux matrix with high-order soliton are given to establish the relationships between inverse scattering method and Darboux transformation. The asymptotic analysis with single spectral parameter is derived through the formulas of determinant directly. Furthermore, the long-time asymptotics with $k$ spectral parameters is given by combining the iterated Darboux matrix and the result of high-order soliton with single spectral parameter, which discloses the structure of high-order soliton clearly and is possible to be utilized in the optic experiments.


\end{abstract}

\date{\today}

\maketitle

\section{Introduction}
We are concerned with the following Hirota equation,
\begin{equation}\label{eq:Hequation}
\ii q_t+\gamma\left(q_{xx}+2|q|^2q\right)+\ii\delta\left(q_{xxx}+6|q|^2q_x\right)=0,
\end{equation}
which was first derived by Hirota \cite{hirota-JMP-1973}. It can be considered as the modified nonlinear Schr\"{o}dinger (NLS) equation with high-order dispersion and time-delay corrections to the cubic nonlinearity. For $\gamma=1$ and $\delta=0$, Eq.\eqref{eq:Hequation} can be reduced to the NLS equation \cite{Hasegawa-book,Agrawal-book}, which is a universal model that describes the evolution of slowly varying packets of quasi-monochromatic waves in weakly nonlinear media. If $\gamma=0$ and $\delta=1$, Eq. \eqref{eq:Hequation}  will turn into the complex modified Korteweg-de Vries equation\cite{Wadati-JPSJ-1973}. Due to the effect of high-order dispersion, Hirota equation \eqref{eq:Hequation} has better properties than NLS equation in some aspect. From a physical view point, in the nonlinear optics, the NLS eqution is an applicable model to describe the picosecond short pulse. While for the subpicosecond or femtosecond pulse, the NLS equation is no longer effective. The high-order dispersion and other effects should be taken into account. In 1985 \cite{Kodama-JSP-1985} and 1987 \cite{Kodama-IEEE-1987}, Hasegawa and Kodama proposed the high-order NLS equation through the Maxwell-equation by the asymptotic perturbation technique in the normalized form (Eq.3.20 and Eq 4.11):
\begin{equation}
\ii Q_{Z}+\frac{1}{2}Q_{TT}+|Q|^2Q+\epsilon\beta_1 \ii \left(Q_{TTT}+6 |Q|^2Q_T\right)=O(\epsilon^2),
\end{equation}
which is derived from another nonlinear equation with higher-order terms
\begin{equation}\label{eq:De-Hirota}
\ii q_{Z}+\frac{1}{2} q_{TT}+|q|^2q+\epsilon\ii\left\{\beta_1 q_{TTT}+\beta_2\left(|q|^2q\right)_{T}+\beta_3 q\left|q^2\right|_{T}\right\}=O(\epsilon^2),
\end{equation}
where $\epsilon$ is small, the function $Q(Z, T)$ and $q(Z, T)$ satisfy
\begin{equation}
Q=q-\epsilon\ii\left(3\beta_1-\frac{1}{2}\beta_2\right)q_{T}-\epsilon\ii\left(6\beta_1-2\beta_2-\beta_3\right)q\int_{-\infty}^{T}\left|q(T')\right|^2dT'+O(\epsilon^2).
\end{equation}
The first term in the brackets of Eq.\eqref{eq:De-Hirota} $\beta_1q_{TTT}$ is the third-order dispersion, the second term $\beta_2\left(|q|^2q\right)_{T}$ indicates the self-steepening effect and the last one $\beta_3q\left|q^2\right|_T$ is the Raman term.  

Although these three high-order terms are small due to the small perturbation $\epsilon$, under the femtosecond pulses, these coefficients can not be neglected any more. \ys{In \cite{Gordon-OL-1986,Mitschke-OL-1986}, the authors analyzed the self-frequency shift of the soliton in theory and experiment, and pointed out that the Raman effects can cause a continuous downshift of the mean frequency of pulses propagating. As a result, the energy will transfer from the higher to the lower-frequency parts of its spectrum. Under the 120-fsec pulses, they observed that the net frequency shift can receive as greater as 10$\%$ of the optical frequency. These experimental results confirm the fact that under the femtosecond regime, the high order terms of Hirota equation are necessary to capture effect that can not be described by NLS equation any more.}

Under some special condition, Eq.\eqref{eq:Hequation} can describe the ultra-short pulse propagation and the phenomenon of ocean waves more precisely than NLS equation.  From the mathematical view point, NLS equation has five kinds of symmetries, while Hirota equation just has four types, that is
\begin{equation}\label{eq:symmetry-equation}
\left\{\begin{split}
\tilde{q}_{1}(x,t)&=\ee^{\ii\epsilon}q(x,t),\\
\tilde{q}_{2}(x,t)&=q(x+\epsilon, t),\\
\tilde{q}_{3}(x,t)&=q(x, t+\epsilon),\\
\tilde{q}_{4}(x,t)&=\ee^{\epsilon-\ii\gamma\left(\frac{ \left(\ee^{\epsilon}-1\right)}{3\delta}x+\frac{\gamma^2\left(\ee^{\epsilon}-1\right)^2\left(\ee^{\epsilon}+2\right)}{27\delta^2}t\right)}
q(\ee^{\epsilon}(x-{\gamma^2(1-\ee^{2\epsilon})}/{(3\delta)}t), \ee^{3\epsilon}t).
\end{split}\right.
\end{equation}
The first three symmetries are consistent with the NLS equation. The fourth one corresponds to the Galilean symmetry, whose parameter $\delta$ is involved in the denominator and can not be equal to zero, so it can not degenerate to the symmetry of the NLS equation. In addition, it is clear that the NLS equation has the scaling symmetry $\tilde{q}(x,t)=\ee^{\epsilon}q(x\ee^{\epsilon}, t\ee^{2\epsilon})$, but the Hirota equation has no scaling invariance. \ys{Compared with the NLS equation, the Hirota equation can be regarded as a $\mathcal{PT}$-symmetric extension of the NLS equation\cite{Cen-PhysD-2019}, with an anti-linear map $\mathcal{PT}$: $x\to -x, t\to-t, \ii\to-\ii, q\to q$. Furthermore, in \cite{Cen-JMP-2019}, the authors studied the solutions to several nonlocal Hirota equations by Hirota bilinear method and Darboux transformation method, which involves six types of choice between $q$ and $r$, $r=q(x, -t), r=q(-x, t), r=q(-x,-t), r=q^*(-x, t), r=q^*(x, -t), r=q^*(-x, -t)$, where $q$ and $r$ are two potentials appearing in the Lax pair. As to the standard Hirota equation, $q$ and $r$ are related by $r=q^*(x,t). $ } Additionally, the soliton in Hirota equation has the different velocity to the NLS due to the high-order dispersion and time-delay corrections. So the properties of Hirota equation will be different from the NLS equation in this aspect.

As an integrable physical model, the soliton solution has always been a research hot. In self modulation problem, the soliton \cite{zakharov-SPJ-1972} is a single wave packet propagating without distortion of its envelope, which is regarded as a balance between the dispersion and the nonlinear term. If two or more solitons admit the distinct velocities, then these solitons will interact each other without changing the shape but with the shifting of phase. Otherwise, the solitons will come into being the breathing or periodic effect\cite{Okamawari-PRA-1995}. If we take the periodic parameter to infinity, the breather will turn into the multi-pole (i.e. high-order) soliton.  Actually, the systematic analysis for the high-order soliton of NLS equation is a long-standing problem in the integrable theory, which was solved very recently \cite{LiBS-17,Schiebold17}.

In the terminology of inverse scattering transform \cite{GGKM-PRL-1967}, $N$-solitons are given in the reflectionless case under the transmission scattering data $a(\lambda)$ has $N$ distinct simple poles. The multi-soliton with different velocities will break up into individual solitons as $t\to\pm\infty$, so it can be used to describe the interaction between $N$ solitons. Especially, when $N\to\infty$, there are infinity solitons, which was first provided by Zhou \cite{Zhou-CPAM-1990}. The explicit infinity solitons to NLS equation and KdV equation were obtained in the reference \cite{Kamvissis-JMP-1995}  and \cite{Gesztesy-Duke-1992} respectively. Especially, in\cite{Peter-CPAM-2007}, the authors gave a detailed analysis to the NLS equation when $N$ is large. On the other hand, the multi-pole soliton are related to the scattering data $a(\lambda)$ with multiple poles, which can describe the interaction between $N$ solitons of equal amplitude but having a particular chirp because they depend on a single spectrum. When $t\to\pm\infty$, they also will break up into $N$ individual solitons.

There will not exist multi-pole nonsingular soliton to the Korteweg-de Vries equation, because the spectral point is simple and its Lax operator is self-adjoint \cite{Carl-DMV-2000}. However, the focusing NLS equation, it allows multiple poles, which was first observed by Zakharov and Shabat \cite{zakharov-SPJ-1972}. Since then, there are many relative multiple pole solutions about various equations, such as the sine-Gordon equation \cite{Poppe-PhysD-1983,Tsuru-JPSJ-1984}, the modified Korteweg-de Vries equation \cite{Wadati-JPSJ-1981}. Especially, Olmedilla\cite{Olmedilla-PhysD-1987} studied the multiple pole solutions to NLS equation by solving Gelfand-Levitan-Marchenko equation and gave the asymptotic analysis to double pole and triple pole soliton when $t\to\infty$, but he left over a conjecture about the asymptotics with general multi-pole soliton. Fortunately, Schiebold confirmed this conjecture and presented a detailed description about the asymptotic analysis \cite{Schiebold17}. Apart from the asymptotics to NLS equation, there are some relative research on the asymptotic behavior on both the (1+1)-dimensional integrable system and (2+1)-dimensional system \cite{Biondini-JMP-2006,Chakravarty-JPA-2008,Matveev-JMP-1994,He-CPL-2019,He-PhyD-2019}. Recently, the multiple pole soliton concept is extended to some nonlocal system and high-order (2+1)-dimensional system\cite{Lou-CTP-2020,Lou-CNSNS-2020}.

\lm{It should be emphasized that, the double pole (second order) soliton was first given in 1972. But until in 2017, the rigorous asymptotitc analysis for multiple pole (high-order) solutions was still an open problem because the formulas are the mixture of exponential and polynomial functions. 
In 1987, Olmedilla put forward a conjecture about this asymptotics via the inverse scattering method. He obtained the asymptotic behavior for lower order $L\leq 9$ due to a complicated calculation. 
In 2017, Schiebold confirmed this conjecture through a method based on the Jordan matrix.
The history manifests that the rigorous asymptotic analysis for the high-order soliton (or multi-high-order soliton) is a long-standing problem in the integrable system.}

Thus it is natural to consider the long-time asymptotics of high-order soliton to the Hirota equation. In \cite{Cen-PhysD-2019,Xu-JPSJ-2020}, the authors gave the asymptotic expression of the second-order and third-order soliton for Hirota equation, but they did not give the asymptotic analysis to general high-order soliton. Compared with the soliton in NLS equation, the soliton in Hirota equation has the different velocity due to the high-order dispersion and time-delay corrections.
So the multi-soliton for NLS equation and Hirota equation is similar. But for the high-order soliton, they are different because the expression is the mixture of exponential and rational functions. And the dynamics of high-order soliton between NLS equation and Hirota equation has a great discrepancy. Thus the long-time asymptotic analysis for the Hirota equation is necessity to further study for the theory of integrable models. We plan to analyze the soliton solution directly to obtain the long-time asymptotics, which is not depending on the Riemann-Hilbert problem(RHP). 
Before this analysis, we need to construct the exact soliton with the aid of Darboux transformation. Now we give a brief introduction about this method.

Darboux transformation is an algebraic method, which has several different versions \cite{GuHZ,Matveev-Salle}. One of the rigorous method is using the loop group \cite{TerngU-00}, which solves the Lax pair by assumption of the holomorphic property of wave functions with $\Phi(\lambda;0,0)=\mathbb{I}$. For the AKNS system with the ${\rm su}(2)$ symmetry, the elementary Darboux matrix can be rewritten as \cite{TerngU-00}:
\begin{equation}
\mathbf{T}(\lambda;x,t)=\mathbb{I}-\frac{\lambda_1-\lambda_1^*}{\lambda-\lambda_1^*}\mathbf{P}(x,t),\qquad \mathbf{P}(x,t)=\frac{\phi_1\phi_1^{\dag}}{\phi_1^{\dag}\phi_1},\qquad \phi_1=\Phi(\lambda_1;x,t)\mathbf{v}_1,
\end{equation}
where $\mathbf{v}_1$ is a constant vector.
And the new wave function $\Phi^{[1]}(\lambda;x,t)=\mathbf{T}(\lambda;x,t)\Phi(\lambda;x,t)\mathbf{T}^{-1}(
\lambda;0,0)$ will satisfy a new Lax pair with the same shape.
The corresponding B\"acklund transformation between old and new potential functions is \begin{equation}
q^{[1]}(x,t)=q(x,t)-2(\lambda_1-\lambda_1^*)\frac{\phi_{1,2}\phi_{1,1}^*}{\phi_1^{\dag}\phi_1},\qquad \phi_1=(\phi_{1,1},\phi_{1,2})^{\T}.
\end{equation}
If $|q(0,0)|=\max(|q(x,t)|)$, the module of new solution $q^{[1]}(x,t)$
will attain the maximal value at $(x,t)=(0,0)$ with $|q^{[1]}(0,0)|=|q(0,0)|+2{\rm Im}(\lambda_1)$ by taking $\mathbf{v}_1=[1,\ii]^{\T}$. This simple proposition is beneficial to construct the solitons with maximal peak. If we want to iterate the Darboux transformation at $\lambda=\lambda_1$, the new wave function $\Phi^{[1]}(\lambda;x,t)$ will appear the removable singularity. So we should redefine the Darboux transformation with limit technique so as to get the high-order soliton at the same spectral point. Based on this method, the high-order soliton to derivative NLS equation \cite{Ling-2013-Study}, the Landau-Lifshitz equation \cite{Ling-2014-Study}, the rogue wave to NLS equation\cite{Ling-2012-PRE}, the spatial discrete Hirota equation\cite{Zhu-CPL-2018} are given systematically. Apart from the generalized Darboux transformation method, some other algebraic method, such the KP reduction method can also be used to construct the high-order soliton or rogue wave\cite{Zhang-CTP-2012,Chen-JPSJ-2018}.

Furthermore, after giving the Darboux transformation and the corresponding B\"{a}cklund transformation, we begin to analyze the long-time asymptotics of high-order soliton. There are three steps in the whole analysis procedure. Firstly, we extract the leading order term from the exact high-order soliton solution determinant with single spectral parameter directly, which is a crucial factor to the long-time asymptotics. Then we find an interesting fact that if the high-order soliton moves along a special characteristic curve, then the limit of the high-order can be reduced to a single soliton. Otherwise, it has a vanishing limitation. Therefore, we give the long-time asymptotics on the basis of moving along the characteristic curve. \lm{In contrast to the analysis for the NLS asymptotics, the asymptotics for the Hirota equation becomes more complicated due to the extra polynomial functions about the $x$ and $t$. Take the second order soliton as an instance:
\begin{equation}\label{eq:sec-H}
\begin{split}
q_{H}&=\frac{\left[192\ii\delta\lambda_{1}^{*2}\lambda_{1I}^2t+64\ii\lambda_{1I}^2\lambda_{1}^*\gamma t+8\ii\left(2\lambda_{1I}x-1\right)\lambda_{1I}\right]\ee^{-2\lambda_{1I}\left(x-v_{H}t\right)-2\lambda_{1R}\ii\left(x-v_{1H}t\right)}}{\ee^{-8\lambda_{1I}(x-v_Ht)}+
\left\{p(\delta; x, t)+16\lambda_{1I}^2\left[\left(x+4\lambda_{1R}\gamma t\right)^2+16\lambda_{1I}^2\gamma^2t^2\right]+2\right\}\ee^{-4\lambda_{1I}(x-v_{H}t)}+1}\\
&-\frac{\left[192\ii\delta\lambda_{1}^2\lambda_{1I}^2t+64\ii\lambda_{1I}^2\lambda_1\gamma t+8\ii\left(2\lambda_{1I}x+1\right)\lambda_{1I}\right]\ee^{-6\lambda_{1I}\left(x-v_{H}t\right)-2\lambda_{1R}\ii\left(x-v_{1H}t\right)}}{\ee^{-8\lambda_{1I}(x-v_Ht)}+
\left\{p(\delta; x, t)+16\lambda_{1I}^2\left[\left(x+4\lambda_{1R}\gamma t\right)^2+16\lambda_{1I}^2\gamma^2t^2\right]+2\right\}\ee^{-4\lambda_{1I}(x-v_{H}t)}+1},
\end{split}
\end{equation}
where
\begin{equation*}
\begin{split}
p(\delta; x,t)=768\lambda_{1I}^2|\lambda_1|^2\left(3|\lambda_1|^2\delta^2+2\lambda_{1R}\delta\gamma\right)t^2
+384\lambda_{1I}^2\left(\lambda_{1R}^2-\lambda_{1I}^2\right)\delta xt\\
v_{H}=4\left(\lambda_{1I}^2-3\lambda_{1R}^2\right)\delta-4\lambda_{1R}\gamma,\qquad v_{1H}=4\left(3\lambda_{1I}^2-\lambda_{1R}^2\right)\delta+2\left(\frac{\lambda_{1I}^2}{\lambda_{1R}}-\lambda_{1R}\right)\gamma,
\end{split}
\end{equation*}
and $\lambda_{1R}, \lambda_{1I}$ are the real and imaginary part of $\lambda_1$ respectively. When $\delta=0, \gamma=1$, it reduces to the soliton to NLS equation. It can be seen that the solution for Hirota equation is more complicated than the NLS equation due to the extra terms with the $\delta$. Especially these factors appear in the polynomial function and the exponent term.  The key idea to analyze the asymptotics is trying to balance the polynomial factor and the exponent factor when $t$ is large, thus we have to consider how these new factors will affect the asymptotics. It should be noted that Eq.\eqref{eq:sec-H} is just a simple second order soliton, as the order increases, the high-order soliton solution of Hirota equation is much more complicated than the one of the NLS equation. The previous article \cite{Olmedilla-PhysD-1987} reported that if order is larger than $9$, the asymptotic behavior about the NLS equation is hard to calculate, let alone the Hirota equation. Thus a systematic tool or method to analyze the high-order soliton replacing  with the tedious calculation is urgently needed for distinct methods in the integrable system. As we know that, the different methods (Hirota method, inverse scattering method, Darboux transformation and many others) in the integrable method will give the distinct formulas, and the equivalence between the formulas is not a trivial work. In our work, we propose the approach to analyze the high-order soliton under the frame of Darboux transformation.
As far as the Hirota equation is concerned,} with the aid of the asymptotics of Darboux matrix and the result of the asymptotics with single spectral parameter, we give a general long-time asymptotics of high-order soliton with $k$ spectral parameters, which has never been reported before.

The outline of this paper is organized as follows. In section \ref{sec2}, the Hirota equation is derived from the AKNS hierarchy. In section \ref{sec3}, the Darboux matrix involving $k$ spectral parameters $\lambda_1, \lambda_2, \cdots, \lambda_k$ of the Hirota equation is given. With the aid of B\"acklund transformation, the high-order soliton could be constructed by taking the zero seed solution. Afterwards, the RHP for the Darboux matrix is derived to establish the relationship between two approaches. In section \ref{sec4}, the asymptotic analysis of high-order soliton with single spectral parameter $\lambda_1$ is shown by the Theorem \ref{theorem1}. Then the long-time asymptotic analysis can be extended to the general case with $k$ spectral parameters, which is given by the Theorem \ref{theorem-k}. The conclusions and discussions are involved in the final section. \lm{This paper contains one appendix as the complement material, which gives some auxiliary proof for the theorem \ref{theorem1}}.

\section{Darboux transformation and the corresponding Riemann-Hilbert representation}
\lm{Our goal in this paper is analyzing the asymptotic behavior when $t$ is large. In contrast to the traditional inverse scattering transformation method, we want to analyze it through the leading order term  from the exact high-order soliton solution. As a first major step towards this goal, we should give the Darboux transformation and the corresponding B\"{a}cklund transformation to construct the high-order soliton. Before the detailed analysis, we first give some prior knowledge in the following:}
\subsection{Integrable hierarchy for the AKNS system and Hirota equation}\label{sec2}
It is well known the Hirota equation\eqref{eq:Hequation} is a mixture of second order and third order flow from the classical AKNS system. Considering the following lax pair
\begin{equation}\label{eq:akns}
\frac{\partial}{\partial x}\pmb{\Phi}(\lambda;x,\mathbf{t})=\mathbf{U}(\lambda;x,\mathbf{t})\pmb{\Phi}(\lambda;x,\mathbf{t}),\qquad \mathbf{U}(\lambda;x,\mathbf{t})=\ii\left(\lambda \sigma_3+\mathbf{Q}(x,\mathbf{t})\right),\qquad \mathbf{Q}(x,\mathbf{t})=\begin{bmatrix}
0 & r(x,\mathbf{t})\\
q(x,\mathbf{t}) &0 \\
\end{bmatrix},
\end{equation}
where $x\in\mathbb{R}$, $\mathbf{t}=(t_1,t_2,\cdots)\in\mathbb{R}^\infty$, $\sigma_3={\rm diag}(1,-1)$ is the third Pauli matrix and $\Phi(\lambda;x,\mathbf{t})$ is the wave function. Now we introduce the real infinite-dimensional iso-spectral manifold \cite{AdlerM-94} $\mathscr{U}$ for the AKNS spectral problem \eqref{eq:akns}. The flow was defined by the following iso-spectral equations
\begin{equation}\label{eq:evolution}
\frac{\partial}{\partial t_i}\pmb{\Phi}(\lambda;x,\mathbf{t})=\ii\left(\lambda^{i}\mathbf{m}(\lambda;x,\mathbf{t})\sigma_3\mathbf{m}^{-1}(\lambda;x,\mathbf{t})\right)_+\pmb{\Phi}(\lambda;x,\mathbf{t}),
\end{equation}
where the subscript $_+$ denotes the positive power of polynomial with respect to $\lambda$, and
\begin{equation}\label{eq:ansatz}
\pmb{\Phi}(\lambda;x,\mathbf{t})=\mathbf{m}(\lambda;x,\mathbf{t})\exp\left[\ii\sigma_3 \left(\lambda(x+t_1)+\sum_{i=2}^{+\infty}t_i\lambda^i \right)\right],\qquad \mathbf{m}(\lambda;x,\mathbf{t})=\mathbb{I}+\mathbf{m}_1(x,\mathbf{t})\lambda^{-1}+\cdots,
\end{equation}
for $\lambda$ in the neighborhood of $\infty$.  Plugging the ansatz \eqref{eq:ansatz} into the spectral problem \eqref{eq:akns}, we have
\lm{
\begin{equation}\label{eq:akns-m}
\ii\lambda \mathbf{m}(x,\mathbf{t};\lambda)\sigma_3\mathbf{m}^{-1}(x,\mathbf{t};\lambda)+\frac{\partial}{\partial x}\mathbf{m}(x,\mathbf{t};\lambda)\mathbf{m}^{-1}(x,\mathbf{t};\lambda)=\ii\left(\lambda\sigma_3+\mathbf{Q}(x,\mathbf{t})\right).
\end{equation}}
Comparing the coefficients of \eqref{eq:akns-m} with respect to $\lambda$, we arrive at
\lm{
\begin{equation}\label{eq:m-q}
\mathbf{Q}(x;\mathbf{t})=\left[\mathbf{m}_1(x,\mathbf{t}),\sigma_3\right],
\end{equation}}
which is useful to reconstruct the solution of integrable equations, where the commutator $[A,B]=AB-BA$. Furthermore, we show that the holomorphic function of $\lambda$: $\mathbf{A}(\lambda;x,\mathbf{t})=\mathbf{m}(\lambda;x,\mathbf{t})\sigma_3\mathbf{m}^{-1}(\lambda;x,\mathbf{t})$ can be expanded in the neighborhood of $\infty$ with the form
\begin{equation}\label{eq:Aexpand}
\mathbf{A}(\lambda;x,\mathbf{t})=\sum_{i=0}^{\infty}\mathbf{A}_i(x,\mathbf{t})\lambda^{-i},\,\,\,\,\, \mathbf{A}_0=\sigma_3,
\end{equation}
where the matrices $\mathbf{A}_i(x,\mathbf{t})$ can be determined by the recursion relationship
\begin{equation}\label{eq:recursion}
\frac{\partial}{\partial x}\mathbf{A}_i(x,\mathbf{t})=\ii[\sigma_3,\mathbf{A}_{i+1}(x,\mathbf{t})]+\ii[\mathbf{Q}(x,\mathbf{t}),\mathbf{A}_{i}(x,\mathbf{t})],\qquad i\geq1; \qquad \mathbf{A}_1^{\rm off}=\mathbf{Q},
\end{equation}
which was implied by the stationary zero-curvature equation $\mathbf{A}_x=[\mathbf{U},\mathbf{A}]$. Taking the off-diagonal part of equation \eqref{eq:recursion}, it will reduce into
\begin{equation}\label{eq:recursion-off}
\mathbf{A}_{i+1}^{\rm off}=-\frac{1}{2}\sigma_3\left(\ii\frac{\partial}{\partial x}\mathbf{A}_i^{\rm off}+\left[\mathbf{Q},\mathbf{A}_i^{\rm diag}\right]\right).
\end{equation} On the other hand, due to a simple identity $\mathbf{m}(\lambda;x,\mathbf{t})[\sigma_3^2-\mathbb{I}]\mathbf{m}^{-1}(\lambda;x,\mathbf{t})=\mathbf{A}^2(\lambda;x,\mathbf{t})-\mathbb{I}=0$, comparing the coefficient of $\lambda$, we give the diagonal elements of $\mathbf{A}_{i}$ as
\begin{equation}\label{eq:coeff-indentity}
\mathbf{A}_{i+1}^{\rm diag}=-\frac{1}{2}\sigma_3\sum_{k=1}^{i}\left(\mathbf{A}_{k}\mathbf{A}_{i+1-k}\right)^{\rm diag},\qquad i\geq 1;\qquad \mathbf{A}_1^{\rm diag}=0.
\end{equation}
Due to the recursive equations \eqref{eq:recursion-off} and \eqref{eq:coeff-indentity}, then we have the following proposition:
\begin{prop}
The matrices $\mathbf{A}_i$ in the equation \eqref{eq:Aexpand} can be represented as the differential polynomials of $\mathbf{Q}$ with the derivative on the variable $x$.
\end{prop}
The first three elements can be represented as
\begin{equation}\label{eq:first-three-flows}
\begin{split}
\mathbf{A}_1=&\mathbf{Q}, \\
\mathbf{A}_2=&-\frac{1}{2}\left(\sigma_3\mathbf{Q}^2+\ii\sigma_3 \mathbf{Q}_x\right), \\
\mathbf{A}_3=&-\frac{1}{4}\left(\mathbf{Q}_{xx}+2\mathbf{Q}^3-\ii\mathbf{Q}_x\mathbf{Q}+\ii\mathbf{Q}\mathbf{Q}_x\right). \\
\end{split}
\end{equation}
Furthermore, if we suppose the matrix function $\mathbf{m}(\lambda;x,\mathbf{t})$ is holomorphic in $\mathbb{C}/\Gamma_0$, where $\Gamma_0$ is the contour in the complex plane, then we have the following propositions:
\begin{prop}\label{prop:laxpair}
The wave functions $\pmb{\Phi}(\lambda;x,\mathbf{t})$ define the following evolution equations:
\begin{equation}\label{eq:evolution}
\frac{\partial}{\partial t_k}\pmb{\Phi}(\lambda;x,\mathbf{t})=\ii \left(\sum_{i=0}^{k}\mathbf{A}_i\lambda^{k-i} \right)\pmb{\Phi}(\lambda;x,\mathbf{t}).
\end{equation}
\end{prop}
\begin{proof}
Since we assume that the matrix function $\mathbf{m}(\lambda;x,\mathbf{t})$ is holomorphic in $\mathbb{C}/\Gamma_0$, then the matrix function $\pmb{\Phi}(\lambda;x,\mathbf{t})=\mathbf{m}(\lambda;x,\mathbf{t})\exp\left[\ii\sigma_3 \left(\lambda(x+t_1)+\sum_{i=2}^{+\infty}t_i\lambda^i \right)\right]$ is also analytic in the complex region $\mathbb{C}/\Gamma_0$. On the other hand, the matrix function $\pmb{\Phi}(\lambda;x,\mathbf{t})$ satisfies the second order differential equation with respect to the variable $x$, then the non-tangential limit of $\pmb{\Phi}(\lambda;x,\mathbf{t})$ on the contour $\Gamma_0$ to the different sides will satisfy the relation $\pmb{\Phi}_+(\lambda;x,\mathbf{t})=\pmb{\Phi}_-(\lambda;x,\mathbf{t})\mathbf{V}(\lambda)$, which implies that the derivative of $\pmb{\Phi}_{\pm}$ with respect to the variable $t_k$ also satisfy the following same jump relationship $\frac{\partial}{\partial t_k}\pmb{\Phi}_+(\lambda;x,\mathbf{t})=\frac{\partial}{\partial t_k}\pmb{\Phi}_-(\lambda;x,\mathbf{t})\mathbf{V}(\lambda)$. Then we know that \[\left(\frac{\partial}{\partial t_k}\pmb{\Phi}_+(\lambda;x,\mathbf{t})\right)\pmb{\Phi}_+^{-1}(\lambda;x,\mathbf{t})=\left(\frac{\partial}{\partial t_k}\pmb{\Phi}_-(\lambda;x,\mathbf{t})\right)\pmb{\Phi}_-^{-1}(\lambda;x,\mathbf{t})\]
on the whole complex plane $\mathbb{C}$, which implies the matrix function $\left(\frac{\partial}{\partial t_k}\pmb{\Phi}(\lambda;x,\mathbf{t})\right)\pmb{\Phi}^{-1}(\lambda;x,\mathbf{t})$
is holomorphic in $\mathbb{C}$. By the calculations, we give the asymptotics in the neighborhood of $\infty$:
\begin{equation}
\begin{split}
\left(\frac{\partial}{\partial t_k}\pmb{\Phi}(\lambda;x,\mathbf{t})\right)\pmb{\Phi}^{-1}(\lambda;x,\mathbf{t})=&
\left(\frac{\partial}{\partial t_k}\mathbf{m}(\lambda;x,\mathbf{t})\right) \mathbf{m}^{-1}(\lambda;x,\mathbf{t})+\ii\lambda^{k}\mathbf{m}(\lambda;x,\mathbf{t})\sigma_3\mathbf{m}^{-1}(\lambda;x,\mathbf{t})\\
=&\left(\ii\lambda^{k}\mathbf{m}(\lambda;x,\mathbf{t})\sigma_3\mathbf{m}^{-1}(\lambda;x,\mathbf{t})\right)_+\\
=&\ii \left(\sum_{i=0}^{k}\mathbf{A}_i\lambda^{k-i} \right).
\end{split}
\end{equation}
By virtue of the Liouville theorem, we get the equations \eqref{eq:evolution}.
\end{proof}

The symmetric reduction for the AKNS spectral problem \eqref{eq:akns} with $q(x,\mathbf{t})=r^*(x,\mathbf{t})$ can be converted into the coefficient matrix through the following ${\rm su}(2)$ reality condition:
\begin{equation}\label{eq:reality}
\mathbf{U}^{\dag}(\lambda^*;x,\mathbf{t})=-\mathbf{U}(\lambda;x,\mathbf{t}).
\end{equation}
Moreover, it can be converted into the symmetric relationship for the wave function $\pmb{\Phi}(\lambda;x,\mathbf{t})$:
\begin{equation}\label{eq:reality-1}
\pmb{\Phi}^{\dag}(\lambda^*;x,\mathbf{t})=\left[\pmb{\Phi}(\lambda;x,\mathbf{t})\right]^{-1}.
\end{equation}

Actually the Hirota equation\eqref{eq:Hequation} is a mixture of second order and third order flow in the above integrable hierarchy under the reduction $q^*(x,\mathbf{t})=r(x,\mathbf{t})$, $t_2=2\gamma t$, $t_3=4\delta t$, where $t_1$ and $t_i$, $i\geq 4$ can be considered as the free parameters for the solutions. Thus the evolution part of Lax pair for the Hirota equation can be represented as the following form:
\begin{equation}\label{eq:hirota-lax-t}
\frac{\partial}{\partial t}\pmb{\Phi}(\lambda;x,t)=2\gamma \frac{\partial}{\partial t_2}\pmb{\Phi}(\lambda;x,t)+4\delta \frac{\partial}{\partial t_3}\pmb{\Phi}(\lambda;x,t)=2\ii\left[\gamma \left(\sum_{i=0}^{2}\mathbf{A}_i\lambda^{2-i} \right)+2\delta\left(\sum_{i=0}^{3}\mathbf{A}_i\lambda^{3-i} \right)\right]\pmb{\Phi}(\lambda;x,t).
\end{equation}
\lm{Generally speaking, the evolution part of Lax pair for the AKNS hierarchy can be obtained by the AKNS scheme \cite{Ablowitz-study-1974}
\begin{equation}\label{eq:anks-t}
\frac{\partial}{\partial t_{n}}\pmb{\Phi}(\lambda;x,\mathbf{t})=\mathbf{V}(\lambda;x,\mathbf{t})\pmb{\Phi}(\lambda;x,\mathbf{t}), \qquad \mathbf{V}(\lambda;x,\mathbf{t})=\begin{bmatrix}A(\lambda;x,\mathbf{t})&B(\lambda;x,\mathbf{t})\\
C(\lambda;x,\mathbf{t})&-A(\lambda;x,\mathbf{t}),
\end{bmatrix}
\end{equation}
where $A(\lambda;x,\mathbf{t}), B(\lambda;x,\mathbf{t}), C(\lambda;x,\mathbf{t})$ are all polynomial functions with respect to $\lambda$. Suppose these matrices as
\begin{equation}
A(\lambda;x,\mathbf{t})=\sum\limits_{j=0}^{n}A_{j}(x,\mathbf{t})\lambda^j, \qquad B(\lambda;x,\mathbf{t})=\sum_{j=0}^{n}B_{j}(x,\mathbf{t})\lambda^j, \qquad C(\lambda;x,\mathbf{t})=\sum\limits_{j=0}^{n}C_j(x,\mathbf{t})\lambda^j.
\end{equation}
The coefficients of $A_j(x,\mathbf{t}), B_{j}(x,\mathbf{t}), C_{j}(x,\mathbf{t})(j=0, 1, \cdots n)$ can be calculated in turn. In particular, the evolution equation is derived from the equations of the coefficient of $\lambda^0$. By choosing different $n$, we can get different evolution equations. The derivation of Lax pair for Hirota equation is well known for us. The goal for this above derivation is to simplify the following proposition \ref{prop3} in the next section.}
\subsection{Darboux transformation and high-order soliton solutions} \label{sec3}
In this subsection, we prepare to construct the exact high-order soliton by Darboux transformation method. Last subsection we have derived the Hirota equation from AKNS hierarchy, the Darboux transformation for the AKNS system with $su(2)$ symmetry is well known in the literature \cite{GuHZ,Matveev-Salle}. The multi-fold Darboux matrix in the frame of loop group, the spectral parameters are different, was given in reference \cite{TerngU-00}. For the general high-order case, the Darboux matrix was given in the literature \cite{Ling-2012-PRE,Peter-Duke-2019} with the following theorem:
\begin{theorem}\label{thm1}
Suppose we have a smooth solution ${q}\in L^{\infty}(\mathbb{R}^2)\cup C^{\infty}(\mathbb{R}^2)$, and the matrix solution $\Phi(\lambda;x,t)$ is holomorphic in the whole complex plane $\mathbb{C}$, the Darboux transformation
\begin{equation}\label{eq:darboux}
\mathbf{T}_N(\lambda;x,t)=\mathbb{I}+\mathbf{Y}_N\mathbf{M}^{-1}\mathbf{D}\mathbf{Y}_N^{\dag},\qquad \mathbf{M}=\mathbf{X}^{\dag}\mathbf{S}\mathbf{X},
\end{equation}
where
\begin{equation}
\begin{split}
\mathbf{Y}_N =&\left[\Phi_1^{[0]},\Phi_1^{[1]},\cdots,\Phi_1^{[n_1-1]}, \Phi_2^{[0]},\Phi_2^{[1]},\cdots,\Phi_2^{[n_2-1]},\cdots, \Phi_k^{[0]},\Phi_k^{[1]},\cdots,\Phi_k^{[n_k-1]} \right], \\
\mathbf{D}=&\begin{bmatrix}
\mathbf{D}_1 &0 &\cdots &0 \\
0 &\mathbf{D}_2 &\cdots & 0 \\
\vdots &\vdots &\vdots & \vdots \\
0 &0 &\cdots &\mathbf{D}_k\\
\end{bmatrix},\qquad \mathbf{D}_i=\begin{bmatrix}
 \frac{1}{\lambda-\lambda_i^*}&0 &\cdots & 0 \\
\frac{1}{(\lambda-\lambda_i^*)^2}&\frac{1}{\lambda-\lambda_i^*} &\cdots & 0 \\
\vdots &\vdots & &\vdots \\
\frac{1}{(\lambda-\lambda_i^*)^{n_i-1}}&\frac{1}{(\lambda-\lambda_i^*)^{n_i-1}} &\cdots &\frac{1}{\lambda-\lambda_i^*} \\
\end{bmatrix},\\
\mathbf{X}=&\begin{bmatrix}
\mathbf{X}_1 &0 &\cdots &0 \\
0 &\mathbf{X}_2 &\cdots & 0 \\
\vdots &\vdots &\vdots & \vdots \\
0 &0 &\cdots &\mathbf{X}_k\\
\end{bmatrix},\qquad \mathbf{X}_i=\begin{bmatrix}
\Phi_i^{[0]}&\Phi_i^{[1]}&\cdots&\Phi_i^{[n_i-1]}\\
0&\Phi_i^{[0]}&\cdots&\Phi_i^{[n_i-2]} \\
\vdots&\vdots&\ddots&\vdots\\
0&0&\cdots& \Phi_i^{[0]}\\
\end{bmatrix},\qquad
\mathbf{S}=\begin{bmatrix}
\mathbf{S}_{11} & \mathbf{S}_{12} &\cdots &\mathbf{S}_{1k} \\
\mathbf{S}_{21} &\mathbf{S}_{22} &\cdots &\mathbf{S}_{2k} \\
\vdots &\vdots &\ddots & \vdots \\
\mathbf{S}_{k1} &\mathbf{S}_{k2} &\cdots &\mathbf{S}_{kk} \\
\end{bmatrix},\\
\end{split}
\end{equation}
and $\Phi_i^{[k]}=\frac{1}{k!}\left(\frac{\rm d}{{\rm d}\lambda}\right)^{k}\Phi_i(\lambda)|_{\lambda=\lambda_i}$, $\Phi_i(\lambda_i)$ belongs to the one dimension linear space: $$\mathrm{span}\{\text{vector solution for Lax pair at }\lambda=\lambda_i\},$$
and \[ \mathbf{S}_{i,j}=\begin{bmatrix}
\binom{0}{0}\frac{\mathbb{I}_2}{\lambda_i^*-\lambda_j}&\binom{1}{0} \frac{\mathbb{I}_2}{(\lambda_i^*-\lambda_j)^2} &\cdots&\binom{n_j-1}{0}\frac{\mathbb{I}_2}{(\lambda_i^*-\lambda_j)^{n_j}} \\
\binom{1}{1}\frac{(-1)\mathbb{I}_2}{(\lambda_i^*-\lambda_j)^2}& \binom{2}{1}\frac{(-1)\mathbb{I}_2}{(\lambda_i^*-\lambda_j)^3} &\cdots&\binom{n_j}{1}\frac{(-1)\mathbb{I}_2}{(\lambda_i^*-\lambda_j)^{n_j+1}} \\\vdots&\vdots&\ddots&\vdots\\
\binom{n_i-1}{n_i-1}\frac{(-1)^{n_i-1}\mathbb{I}_2}{(\lambda_i^*-\lambda_j)^{n_i}}& \binom{n_i}{n_i-1}\frac{(-1)^{n_i-1}\mathbb{I}_2}{\lambda_i^*-\lambda_j)^{n_i+1}} &\cdots&\binom{n_i+n_j-2}{n_i-1}\frac{(-1)^{n_i-1}\mathbb{I}_2}{(\lambda_i^*-\lambda_j)^{n_i+n_j-1}}\\
\end{bmatrix},\]
converts the Lax pair \eqref{eq:akns} and \eqref{eq:hirota-lax-t} into a new one by replacing the potential functions
\begin{equation}\label{eq:back}
{q}^{[N]}={q}+2 \mathbf{Y}_{N,2}\mathbf{M}^{-1}\mathbf{Y}_{N,1}^{\dag},
\end{equation}
which is the B\"acklund transformation, where the subscript $\mathbf{Y}_{N,2}$ denotes the second row vector of $\mathbf{Y}_N$, and $\mathbf{Y}_{N,1}$ denotes the first row vector of $\mathbf{Y}_N$.
\end{theorem}

Now we proceed to construct the soliton solutions by the above B\"acklund transformation \eqref{eq:back}. We start with a seed solution: $q(x,t)=0,$
solving the Lax pair Eq.\eqref{eq:akns} and Eq.\eqref{eq:hirota-lax-t} gives a vector from the span space of the fundamental solution:
\begin{equation}\label{eq:funsolution}
 \Phi_{i}^{[k]}=\begin{bmatrix}\frac{1}{k!}\left(\frac{d}{d\lambda}\right)^k\ee^{2\theta^{[i]}}\Big|_{\lambda=\lambda_i}\\
\delta_{k,0}
\end{bmatrix},\qquad  \theta^{[i]}:=\ii\lambda(x+(2\gamma\lambda+4\delta\lambda^2)t)+\sum\limits_{j=0}^{n_i-1}a_i^{[j]}(\lambda-\lambda_i)^j-\frac{\ii}{4}\pi,
\end{equation}
where $\delta_{k,0}$ is the standard Chrestoffel symbol.

By the B\"acklund transformation \eqref{eq:back}, the single soliton reads
\begin{equation}\label{eq:one-soliton}
q^{[1]}=2\lambda_{1I}\sech\left(2\lambda_{1I}\left(x-v_1t\right)-2a_1^{[0]}\right)
\ee^{-2\ii{\rm Im}(\theta_1)-\frac{\ii}{2}\pi},
\end{equation}
where
\begin{equation}
\begin{split}
v_1&=-4\left(\lambda_{1R}\gamma-\lambda_{1I}^2\delta+3\lambda_{1R}^2\delta\right), \quad \theta_1:=\theta^{[1]}\big|_{\lambda=\lambda_1}=\ii\lambda_1(x+(2\gamma\lambda_1+4\delta\lambda_1^2)t)
+a_1^{[0]}-\frac{\ii}{4}\pi,\\
\lambda_1&=\lambda_{1R}+\ii\lambda_{1I},\quad {\rm Im}(\theta_1)=\lambda_{1R}x+2\left(\lambda_{1R}^2-\lambda_{1I}^2\right)\gamma t+4\left(\lambda_{1R}^2-3\lambda_{1I}^2\right)\lambda_{1R}\delta t-\frac{\pi}{4},
\end{split}
\end{equation}
and the script $_1$ in $v_1, \theta_1$ stands for $v, \theta^{[i]}$ depending on $\lambda_1$. Clearly, the velocity of single soliton $v_1$ is a quadratic function with respect to $\lambda_{1R}$ and $\lambda_{1I}$, so it has a minimum $v_{min}=\delta\left(4\lambda_{1I}^2+\frac{\gamma^2}{3\delta^2}\right)$.

Besides, with the aid of the fourth symmetry in Eq.\eqref{eq:symmetry-equation}, if $q^{[1]}$ is a solution to Hirota equation \eqref{eq:Hequation}, then
\begin{equation}
\begin{split}
\tilde{q}^{[1]}(x, t)&=2\lambda_{1I}\ee^{\epsilon}\sech\left(2\lambda_{1I}\ee^{\epsilon}\left(x-\tilde{v}_1 t\right)-2a_1^{[0]}\right)
\ee^{-2\ii{\rm Im}(\tilde{\theta}_1)-\frac{\ii}{2}\pi}
\end{split}
\end{equation}
is also the solution, where
\begin{equation*}
\begin{split}
\tilde{v}_1&=v_1\left(\lambda_{1I}\to\lambda_{1I}\ee^{\epsilon}, \quad \lambda_{1R}\to\lambda_{1R}\ee^{\epsilon}+\frac{ \left(\ee^{\epsilon}-1\right)}{6\delta}\gamma\right),\\
{\rm Im}(\tilde{\theta}_1)&={\rm Im}\left(\theta_1\left(\lambda_{1I}\to\lambda_{1I}\ee^{\epsilon}, \quad \lambda_{1R}\to\lambda_{1R}\ee^{\epsilon}+\frac{ \left(\ee^{\epsilon}-1\right)}{6\delta}\gamma\right)\right).
\end{split}
\end{equation*}

Furthermore, we could derive the high-order soliton with the B\"{a}cklund transformation \eqref{eq:back}. By choosing some special parameters, we give some examples to describe the dynamical behavior for high-order soliton. Suppose $N=2$, then there will appear two kinds of soliton, one is $n_1=2$, the other one is $n_1=1, n_2=1$. The first one depends on single spectral parameter $\lambda_1$, and these two solitons have the same amplitude and velocity. The second one depends on two spectral parameters $\lambda_1$ and $\lambda_2$, so their dynamical properties will be different from the first one, their velocities and the amplitudes are decided by $\lambda_1$ and $\lambda_2$. Moreover, we also give some high-order soliton under $N=4$.
\begin{figure}[!h]
{
\includegraphics[height=0.4\textwidth]{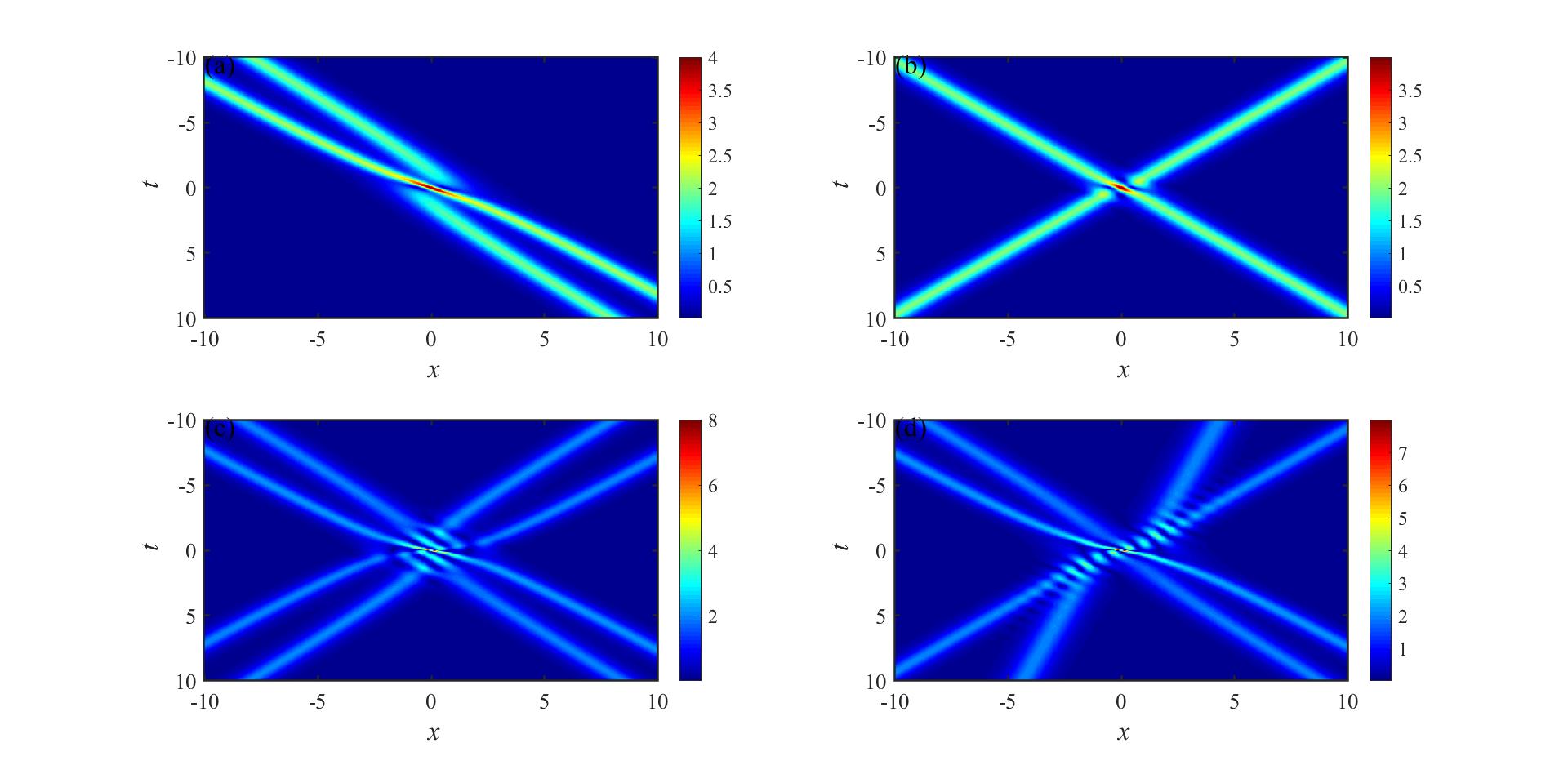}}
\caption{\small Four types of solitons with $\delta=1/8, \gamma=1/2$, (a) is second order solion with $\lambda_1=-1/3+\ii, n_1=2, a_1^{[0]}=a_{1}^{[1]}=0$, (b) is a multi-soliton, with $\lambda_1=-1/3+\ii, \lambda_2=-2/3-\sqrt{13}/3+\ii, n_1=1, n_2=1, a_{1}^{[0]}=a_{2}^{[0]}=0,$ (c) is a high-order soltion with $\lambda_1=-1/3+\ii, \lambda_2=-2/3-\sqrt{13}/3+\ii, n_1=2, n_2=2, a_{1}^{[0]}=a_{1}^{[1]}=a_{2}^{[0]}=a_{2}^{[1]}=0$, (d) is also high-order soliton with $\lambda_1=-1/3+\ii, \lambda_2=-2/3-\sqrt{13}/3+\ii, \lambda_3=1/3+\ii, n_1=2, n_2=1, n_3=1, a_{1}^{[0]}=a_{1}^{[1]}=a_{2}^{[0]}=a_{3}^{[0]}=0.$}\label{fig:higher-order}
\end{figure}

From the pictures, we can see the maximum amplitude in Fig.\ref{fig:higher-order} (a) and (b) is $4$, (c) and (d) is $8$. In fact, from the B\"{a}cklund transformation \eqref{eq:back}, we can calculate the maximum with $k$ spectral parameters $\lambda_1, \lambda_2, \cdots, \lambda_k$ with order $n_1, n_2, \cdots, n_k$ respectively, and the maximum amplitude is $$|q|^{[N]}_{\max}=\sum_{i=1}^{k}2|\lambda_{iI}|n_{i}.$$
\begin{figure}[!h]
{
\includegraphics[height=0.4\textwidth]{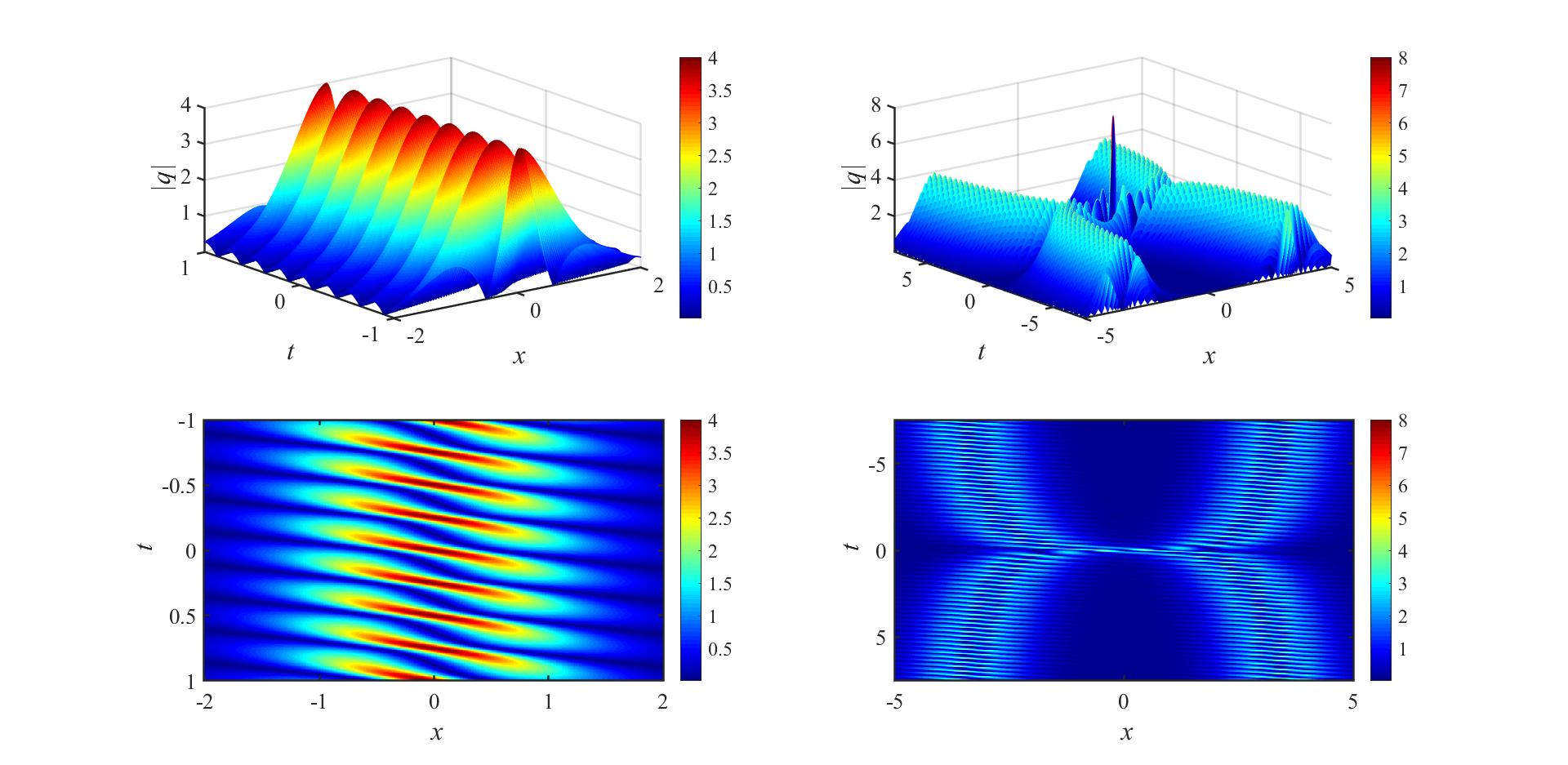}}
\caption{\small Two types of soliton with the same velocity, one is $n_1=1, n_2=1,$ the other is $n_1=2, n_2=2.$ The parameters is $\lambda_1=-2/3+\ii, \lambda_2=1/2+\ii, \delta=1, \gamma=1/2, a_1^{[0]}=a_{1}^{[1]}=a_{2}^{[0]}=a_{2}^{[1]}=0.$}\label{fig:same-v}
\end{figure}

From the figures, if the soliton depends on the same spectral parameter $\lambda_i$, they will move along with the same direction. But if the soliton depends on different spectral parameters, their moving directions can be changed by choosing different spectral values. Especially, if the velocity $v_1=v_2$, then there will appear a more interesting phenomenon, two solitons will come into being a breather, four solitons will become two breathers, which can be seen from Fig.\ref{fig:same-v}. It is clear that the dynamical behavior between the multi-soltion and the high-order soliton is different. However, they have some common features in some aspects, such as when $t\to\infty$, they both break up into individual soliton. Thus what does the long-time asymptotics of high-order soliton will be? How to obtain the asymptotic expression when $t\to\infty$? We would like to disclose the answers in the following.

\subsection{Riemann-Hilbert representation for the Darboux matrix of high-order soliton}

Before the analysis on the high-order soliton, we would like to establish the relationship between the Darboux transformation and Riemann-Hilbert method.
In the literature \cite{Peter-Duke-2019,Deniz-JNS-2019,Deniz-arXiv-2019}, the authors discussed the asymptotic behavior of infinity order of rogue wave and soliton with the robust inverse scattering method. The jump matrix is only related to the Darboux matrix and the phase terms, that is to say, the Darboux matrix can be regarded as a special $\mathbf{M}(\lambda; x, t)$ matrix appeared in the RHP. Inspired by this idea, we construct two sectional analytic matrix
\begin{equation}
\mathbf{M}^{[N]}(\lambda; x, t)=\left\{\begin{split}
\mathbf{M}_{+}^{[N]}(\lambda; x, t)&=\left(\frac{\lambda-\lambda_1}{\lambda-\lambda_1^*}\right)^{-N/2}\mathbf{T}_{N}(\lambda; x, t),\quad \lambda\notin D_0, \\
\mathbf{M}_{-}^{[N]}(\lambda; x, t)&=\mathbf{T}_{N}(\lambda; x, t)\mathbf{\Phi}(\lambda; x, t)\mathbf{T}_{N}^{-1}(\lambda; 0, 0)\mathbf{\Phi}^{-1}(\lambda; x, t),\quad \lambda\in D_0,
\end{split}\right.
\end{equation}
where $D_0$ is a big disk centered at $zero$ involving the parameter $\lambda_1$, then this matrix solve the following RHP:
\begin{rhp}\label{rhp-TM}
Let $(x, t)\in \mathbb{R}^2$ be arbitrary parameters, and let $N\in\mathbb{Z}_{\geq0}$. Find a $2\times 2$ matrix function $\mathbf{M}^{[N]}(\lambda; x, t)$ satisfying the following properties
\begin{itemize}

\item \textbf{Analyticity}: $\mathbf{M}^{[N]}(\lambda; x, t)$ is analytic for $\lambda\in \mathbb{C}\setminus \partial D_{0}$, and it takes continuous boundary values from the interior and exterior of $\partial D_0$.
\item \textbf{Jump condition}: The boundary values on the jump contour $\partial D_0$ are related as
\begin{equation*}
\mathbf{M}^{[N]}_+(\lambda; x, t)=\mathbf{M}_{-}^{[N]}(\lambda; x, t)\ee^{\theta\sigma_3}\mathbf{Q}\left(\frac{\lambda-\lambda_1}{\lambda-\lambda_1^*}\right)^{\frac{N}{2}\sigma_3}
\mathbf{Q}^{-1}\ee^{-\theta\sigma_3},
\end{equation*}
\item \textbf{Normalization}: $\mathbf{M}^{[N]}(\lambda; x, t)=\mathbb{I}+O(\lambda^{-1})$.
\end{itemize}
where $\mathbf{Q}=\frac{1}{\sqrt{2}}\begin{bmatrix}1&\ii\\
\ii&1
\end{bmatrix}, \theta=\ii\lambda(x+(2\gamma\lambda+4\delta\lambda^2)t).$
\end{rhp}
It needs to be emphasized that the RHP constructed here and in the following are all reflectionless, and the Darboux matrix depends on single spectral parameter $\lambda_1$.

Based on the method in \cite{Deniz-JNS-2019,Deniz-arXiv-2019}, we can define the following sectional holomorphic matrix function \begin{equation}
\mathbf{\widetilde{M}}^{[N]}(\lambda; x, t):=\left\{\begin{split}
&\mathbf{M}^{[N]}(\lambda; x, t)\ee^{\theta\sigma_3}\mathbf{Q}\ee^{-\theta\sigma_3}, &\lambda\in D_0\\
&\mathbf{M}^{[N]}(\lambda; x, t)\left(\frac{\lambda-\lambda_1^*}{\lambda-\lambda_1}\right)^{\frac{N}{2}\sigma_3},&\lambda \notin D_0
\end{split}\right.
\end{equation}
which solves the following RHP:
\begin{rhp}\label{rhp-M}
Define $D_0\in \mathbb{C}$ be a disk centered at the origin containing $\lambda_1$ in its interior. Find a unique $2\times 2$ matrix function satisfying the following properties:
\begin{itemize}
\item \textbf{Analyticity}: $\mathbf{\widetilde{M}}^{[N]}(\lambda; x, t)$ is analytic for $\lambda\in \mathbb{C}\setminus \partial D_{0}$, and it takes continuous boundary values from the interior and exterior of $\partial D_0$.
\item \textbf{Jump condition}: The boundary values on the jump contour $\partial D_0$ are related as
\begin{equation*}
\mathbf{\widetilde{M}}^{[N]}_+(\lambda; x, t)=\mathbf{\widetilde{M}}_{-}^{[N]}(\lambda; x, t)\ee^{\left(\theta+ \frac{N}{2}\log\left(\frac{\lambda-\lambda_1}{\lambda-\lambda_1^*}\right)\right)\sigma_3}
\mathbf{Q}^{-1}\ee^{-\left(\theta+ \frac{N}{2}\log\left(\frac{\lambda-\lambda_1}{\lambda-\lambda_1^*}\right)\right)\sigma_3},
\end{equation*}
\item \textbf{Normalization}: $\mathbf{\widetilde{M}}^{[N]}(\lambda; x, t)=\mathbb{I}+O(\lambda^{-1})$.
\end{itemize}
\end{rhp}
In RHP \ref{rhp-TM}, we construct two sectional analytic matrices $\mathbf{M}^{[N]}(\lambda; x, t)$, and the jump matrix exist in a big circle and the parameter is involved in the circle. The classical RHP without reflection coefficients can be constructed by two sectional analytic matrices in and out of two small circle centered at $\lambda_1$ and $\lambda_1^*$, as shown in Fig.\ref{fig:jump}.
\begin{figure}[!h]
{\includegraphics[height=0.35\textwidth]{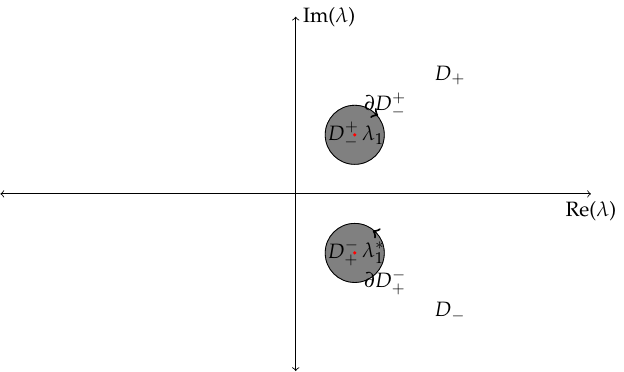}}
\caption{\small Definition of $D_{\pm}, D_{+}^{-}, D_{-}^+$, and the contour $\partial D_{+}^{-}, \partial D_{-}^{+}.$}\label{fig:jump}
\end{figure}
So we can construct another RHP with the jump in Fig.\ref{fig:jump}.
\begin{lemma}\label{lemma-M}
For the Darboux matrix $\mathbf{T}_{N}(\lambda; x, t)$ with $n_1=N$, two analytic matrices can be constructed as follows:
\begin{equation}\label{eq:negative-M}
\mathbf{T}_{N}(\lambda; x, t)\begin{bmatrix}1&0\\0&\left(\frac{\lambda-\lambda_1^*}{\lambda-\lambda_1}\right)^N
\end{bmatrix}\begin{bmatrix}1&-\sum\limits_{i=1}^{N}\frac{\alpha_i}{(\lambda-\lambda_1)^i}\\0&1
\end{bmatrix},
\end{equation}
and
\begin{equation}\label{eq:positive-M}
\mathbf{T}_{N}(\lambda; x, t)\begin{bmatrix}1&0\\0&\left(\frac{\lambda-\lambda_1^*}{\lambda-\lambda_1}\right)^N
\end{bmatrix}\begin{bmatrix}1&0\\\sum\limits_{i=1}^{N}\frac{\alpha_i^*}{(\lambda-\lambda_1^*)^i}&1
\end{bmatrix},
\end{equation}
which are analytic at the neighborhood of $\lambda=\lambda_1$, $\lambda=\lambda_1^*$ respectively, where $\alpha_i=\ii\frac{1}{(N-i)!}\frac{\partial^{N-i} (\lambda-\lambda_1^*)^N\ee^{2\theta}}{\partial\lambda^{N-i}}\big|_{\lambda=\lambda_1}$.
\end{lemma}
\begin{proof}
Through the kernel condition of the Darboux matrix in theorem \ref{thm1}, we know that the meromorphic vector
$\mathbf{T}_N(\lambda;x,t) \left[\ee^{2\theta},\ii\right]^{\T}$ can be expanded with the form $O((\lambda-\lambda_1)^N)$ at the neighborhood of $\lambda=\lambda_1$. Furthermore, we know that
\begin{equation}\label{eq:sing}
\frac{\mathbf{T}_{N}(\lambda; x, t)}{\left(\lambda-\lambda_1\right)^N}\begin{bmatrix}-\ii\left(\lambda-\lambda_1^*\right)^N\ee^{2\theta}\\\left(\lambda-\lambda_1^*\right)^N
\end{bmatrix}\end{equation} is holomorphic at the neighborhood of $\lambda=\lambda_1$. The principal part at $\lambda=\lambda_1$ for the matrix $$ \mathbf{T}_{N}(\lambda; x, t)\begin{bmatrix}1&0\\0&\left(\frac{\lambda-\lambda_1^*}{\lambda-\lambda_1}\right)^N
\end{bmatrix}$$ can be eliminated by the analytic property of \eqref{eq:sing}: i.e. $$\mathbf{T}_{N}(\lambda; x, t)\begin{bmatrix}1&-\frac{\sum\limits_{i=1}^{N}\alpha_i\left(\lambda-\lambda_1\right)^{N-i}}{\left(\lambda-\lambda_1\right)^{N}}\\0&\left(\frac{\lambda-\lambda_1^*}{\lambda-\lambda_1}\right)^N
\end{bmatrix},$$ is analytic at the neighborhood of $\lambda=\lambda_1, $ where $-\ii\left(\lambda-\lambda_1^*\right)^N\ee^{2\theta}=\sum\limits_{i=1}^{N}-\alpha_i(\lambda-\lambda_1)^{N-i}+O((\lambda-\lambda_1)^{N}).$
In a similar manner, the equation \eqref{eq:positive-M} is analytic at the neighborhood of $\lambda=\lambda_1^*$ can be proved.
\end{proof}

With the aid of lemma \ref{lemma-M}, we can construct two kinds of sectional analytic matrix as
\begin{equation}
\widehat{\mathbf{M}}^{[N]}(\lambda; x, t)=
\left\{\begin{split}&\widehat{\mathbf{M}}^{[N]}_{+}(\lambda; x, t)=\mathbf{T}_{N}(\lambda; x, t)\begin{bmatrix}1&0\\
0&\left(\frac{\lambda-\lambda_1^*}{\lambda-\lambda_1}\right)^{N}
\end{bmatrix}, \quad &\lambda\in D_{+}\cup D_-\\
&\widehat{\mathbf{M}}^{[N]}_{-}(\lambda; x, t)=\mathbf{T}_{N}(\lambda; x, t)\begin{bmatrix}1&0\\
0&\left(\frac{\lambda-\lambda_1^*}{\lambda-\lambda_1}\right)^{N}
\end{bmatrix}\begin{bmatrix}
1&-\sum\limits_{i=1}^{N}\frac{\alpha_i}{(\lambda-\lambda_1)^i}\\
0&1
\end{bmatrix},\quad &\lambda\in D^{+}_{-}\\
&\widehat{\mathbf{M}}^{[N]}_{-}(\lambda; x, t)=\mathbf{T}_{N}(\lambda; x, t)\begin{bmatrix}1&0\\
0&\left(\frac{\lambda-\lambda_1^*}{\lambda-\lambda_1}\right)^{N}
\end{bmatrix}\begin{bmatrix}
1&0\\
\sum\limits_{i=1}^{N}\frac{\alpha_i^*}{(\lambda-\lambda_1^*)^i}&1
\end{bmatrix},\quad &\lambda\in D_{+}^-
\end{split}\right.
\end{equation}
which solves the following RHP:
\begin{rhp}\label{rhp-M1}
Define two small circles $D_{-}^{+},D_{+}^{-}$ centered at $\lambda=\lambda_1$ and $\lambda=\lambda_1^*$ respectively. Find a unique $2\times 2$ Matrix $\widehat{\mathbf{M}}^{[N]}(\lambda; x, t)$, satisfying the following properties:
\begin{itemize}

\item \textbf{Analyticity}: $\widehat{\mathbf{M}}^{[N]}(\lambda; x, t)$ is analytic for $\lambda\in \mathbb{C}\setminus \left(\partial D_{-}^+\bigcup\partial D_{+}^-\right)$, and it takes continuous boundary values from the interior and exterior of $D_{+}^{-}$ and $D_{-}^+$.
\item \textbf{Jump condition}: The boundary values on the jump contour $\partial D_{-}^{+}, \partial D_{+}^{-}$ are related as
\begin{equation*}
\begin{split}
&\widehat{\mathbf{M}}^{[N]}_{+}(\lambda; x, t)=\widehat{\mathbf{M}}^{[N]}_{-}(\lambda; x, t)\mathbf{v}_{+}(\lambda; x, t),\qquad\lambda\in\partial D_{-}^{+}\\
&\widehat{\mathbf{M}}^{[N]}_{+}(\lambda; x, t)=\widehat{\mathbf{M}}^{[N]}_{-}(\lambda; x, t)\mathbf{v}_{-}(\lambda; x, t),\qquad\lambda\in\partial D_{+}^{-}
\end{split}
\end{equation*}
where
\begin{equation}\label{eq:alphaj}
\begin{split}
&\mathbf{v}_{+}(\lambda; x, t)=\begin{bmatrix}
1&\sum\limits_{i=1}^{N}\frac{\alpha_i}{(\lambda-\lambda_1)^i}\\
0&1
\end{bmatrix}, \quad \mathbf{v}_{-}(\lambda; x, t)=\begin{bmatrix}
1&0\\
\sum\limits_{i=1}^{N}\frac{\alpha_i^*}{(\lambda-\lambda_1^*)^i}&1
\end{bmatrix},
\end{split}
\end{equation}
\item \textbf{Normalization}: $\widehat{\mathbf{M}}^{[N]}(\lambda; x, t)=\mathbb{I}+O(\lambda^{-1})$.
\end{itemize}
\end{rhp}
Compared with RHP \ref{rhp-M}, the jump matrix in the RHP  \ref{rhp-M1} is different in essence, which is a polynomial function with respect to $x$ and $t$. In the literature \cite{Peter-Duke-2019}, Bilman, the second author of this work and Miller derived the Lax pair from the RHP  \ref{rhp-M}. Here we consider how to derive the Lax pair by the RHP  \ref{rhp-M1}.
\begin{prop}\label{prop3}
The RHP  \ref{rhp-M1} can derive the Lax pair of Hirota equation.
\end{prop}
\begin{proof}
According to the the proposition \ref{prop:laxpair}, we only need to prove $$\left(\widehat{\mathbf{M}}_+(\lambda; x, t)\ee^{\theta\sigma_3}\right)_x\left(\widehat{\mathbf{M}}_+(\lambda; x, t)\ee^{\theta\sigma_3}\right)^{-1}\quad\text{and}\quad\left(\widehat{\mathbf{M}}_+(\lambda; x, t)\ee^{\theta\sigma_3}\right)_t\left(\widehat{\mathbf{M}}_+(\lambda; x, t)\ee^{\theta\sigma_3}\right)^{-1}$$ are holomorphic function in the whole complex plane $\mathbb{C}$. With the jump condition in RHP  \ref{rhp-M1}, in the region $D_-^+$, we have
\begin{equation}
\begin{split}
&\left(\widehat{\mathbf{M}}^{[N]}_{+}(\lambda; x, t)\ee^{\theta\sigma_3}\right)_x\left(\widehat{\mathbf{M}}^{[N]}_{+}(\lambda; x, t)\ee^{\theta\sigma_3}\right)^{-1}
\\=&\left(\widehat{\mathbf{M}}^{[N]}_{-}(\lambda; x, t)\mathbf{v}_{+}(\lambda; x, t)\ee^{\theta\sigma_3}\right)_x\left(\widehat{\mathbf{M}}^{[N]}_{-}(\lambda; x, t)
\mathbf{v}_{+}(\lambda; x, t)\ee^{\theta\sigma_3}\right)^{-1}\\=&
\widehat{\mathbf{M}}^{[N]}_{-}(\lambda; x, t)\left\{\left[\mathbf{v}_{+, x}(\lambda; x, t){+}\mathbf{v}_{+}(\lambda; x, t)(\ii\lambda\sigma_3)\right]\mathbf{v}_{+}^{-1}(\lambda; x, t)\right\}\left(\widehat{\mathbf{M}}_{-}^{[N]}\right)^{-1}(\lambda; x, t)\\+&\widehat{\mathbf{M}}^{[N]}_{-, x}(\lambda; x, t)(\widehat{\mathbf{M}}^{[N]}_{-})^{-1}(\lambda; x, t).
\end{split}
\end{equation}
Furthermore, with the expression of $\mathbf{v}_{+}(\lambda; x, t)$, we know
\begin{equation}
\begin{split}
\left[\mathbf{v}_{+, x}(\lambda; x, t){+}\mathbf{v}_{+}(\lambda; x, t)(\ii\lambda\sigma_3)\right]\mathbf{v}_{+}^{-1}(\lambda; x, t)=
\begin{bmatrix}
\ii\lambda&\left(\sum\limits_{i=1}^{N}\frac{\alpha_i}{(\lambda-\lambda_1)^i}\right)_x
-2\ii\lambda\left(\sum\limits_{i=1}^{N}\frac{\alpha_i}{(\lambda-\lambda_1)^i}\right)\\
0&-\ii\lambda
\end{bmatrix}.
\end{split}
\end{equation}
Based on the definition of $\alpha_j$ in Lemma \ref{lemma-M}, we know that the $(1,2)$ element of the matrix at the right hand side in the above equation is analytic in the neighborhood of $\lambda=\lambda_1$, which deduces that the matrix $\left(\widehat{\mathbf{M}}^{[N]}_{+}(\lambda; x, t)\ee^{\theta\sigma_3}\right)_x\left(\widehat{\mathbf{M}}^{[N]}_{+}(\lambda; x, t)\ee^{\theta\sigma_3}\right)^{-1}$ is analytic in the region $D_-^+$. Similarly, we can prove that $\left(\widehat{\mathbf{M}}^{[N]}_{+}(\lambda; x, t)\ee^{\theta\sigma_3}\right)_x\left(\widehat{\mathbf{M}}^{[N]}_{+}(\lambda; x, t)\ee^{\theta\sigma_3}\right)^{-1}$ is analytic in the region $D_+^-$. Therefore, we conclude that the matrix $\left(\widehat{\mathbf{M}}^{[N]}_{+}(\lambda; x, t)\ee^{\theta\sigma_3}\right)_x\left(\widehat{\mathbf{M}}^{[N]}_{+}(\lambda; x, t)\ee^{\theta\sigma_3}\right)^{-1}$ is holomorphic on the whole complex plane $\mathbb{C}$.

In a similar manner, we consider the $t$-part. By the similar property of $x$-part, we have \begin{equation}
\begin{split}
&\left[\mathbf{v}_{+, t}(\lambda; x, t){+}\mathbf{v}_{+}(\lambda; x, t)\left(2\ii\lambda^2(\gamma+2\delta\lambda)\sigma_3\right)\right]\mathbf{v}_{+}^{-1}(\lambda; x, t)\\=&
\begin{bmatrix}
2\ii\lambda^2(\gamma+2\delta\lambda)&\left(\sum\limits_{i=1}^{N}\frac{\alpha_i}{(\lambda-\lambda_1)^i}\right)_t
-4\ii\lambda^2(\gamma+2\delta\lambda)\left(\sum\limits_{i=1}^{N}\frac{\alpha_i}{(\lambda-\lambda_1)^i}\right)\\
0&-2\ii\lambda^2(\gamma+2\delta\lambda)
\end{bmatrix},
\end{split}
\end{equation}
which implies that $\left(\widehat{\mathbf{M}}_{+}(\lambda; x, t)\ee^{\theta\sigma_3}\right)_t\left(\widehat{\mathbf{M}}_{+}(\lambda; x, t)\ee^{\theta\sigma_3}\right)^{-1}$ is analytic in the region $D_-^+$. With the same method, we can also prove that it is analytic in the region $D_+^-$. Then the Lax pair of Hirota equation can be derived by the proposition \ref{prop:laxpair}.
\end{proof}
\lm{These two kinds of RHP both are constructed from the given Darboux matrix in Eq.\eqref{eq:darboux} directly, which produces a new way to establish the relation between the Darboux transformation and the Riemann-Hilbert problem. Namely, under the reflectionless case, we can give a RHP as long as we have a Darboux matrix, which did not reported before to the best of our knowledge}. Moreover, this new Riemann-Hilbert representation will be useful in the analysis to the large order soliton \cite{Deniz-JNS-2019,Deniz-arXiv-2019}. The further analysis on the large order soliton in different region will be another topic, whose difficulty is that the phase term $\theta$ is a cubic polynomial with respect to $\lambda$.

Although we give two kinds of RHP of high-order soliton, we intend to analyze the long-time asymptotics from the solutions determinant directly. Observing the Fig.\ref{fig:higher-order}, when $t\to\infty$, Fig.\ref{fig:higher-order} (a) and (b) has two individual solitons, and (c) and (d) has four individual solitons. We hope to get the asymptotic individual solitons with the form of single soliton. That is, when $t\to\infty$, if there are $N$ individual solitons, we should construct $N$ single solitons as the asymptotic expression for $N$-th order soliton.
\section{The asymptotic analysis to high-order soliton} \label{sec4}
\lm{Last section, we exhibit some graphics on the high-order soltion. Now we are going to proceed to analyze the long-time asymptotic behavior about the high-order soliton. From the figures, we can see that the $N$-th order soliton can break up into $N$ individual solitons as $t\to\infty$. We would like to give a strict proof on this issue. So how to find the $N$ individual solitons to match the $N$-th order soliton?  In order to tackle with this problem, we first gives the asymptotic expression under one spectral parameter $\lambda_1$. To realize it, we extract the leading order term from the high-order soliton determinant given in the last section \eqref{eq:back}, which contains polynomial function and exponent function. If both two kinds of function have the same degree $t$, then the long-time asymptotics can be expressed by a single soliton. Otherwise, it is a vanishing solution. Afterwards, the general asymptotic behavior with $k$ spectral parameters can be given on the basis of the former result. }
\subsection{The asymptotic behavior under one spectral parameter}
\lm{In this subsection, we prepare to give the asymptotic behavior with only one spectral parameter $\lambda_1$. Under this condition, the high-order soliton in Eq.\eqref{eq:back} seems easier than multi-parameters case. We will try to get the leading order term when $t$ is large for further study. During the calculation, we need a lot of determinant techniques to obtain the leading order terms. To give a detailed analysis, we gives some prior lemmas in the following:}
\begin{lemma}\label{lemma1}
Denote $\vartheta_{1,0}=-2\left[x+4\lambda_1(\gamma+3\delta\lambda_1)t\right]+2\ii a_{1}^{[1]}$, $\vartheta_{1,1}=-4\ii(\gamma+6\delta\lambda_1)t-2a_{1}^{[2]}$, $\vartheta_{1,2}=8\delta t-2\ii a_{1}^{[3]}$, then we have the following expansion for the function
\begin{equation}\label{eq:expand-exp}
\exp(\epsilon(\vartheta_{1,0}+\epsilon\vartheta_{1,1}+\epsilon^2\vartheta_{1,2}))=1+\sum_{i=1}^{\infty}f_{1,l}\epsilon^l,
\end{equation}
where
\begin{equation}\label{eq:f-exp}
f_{1,l}=\sum_{m=0}^{[2l/3]}\sum_{k_0=\max(l-2m,0)}^{[l-3m/2]}\frac{\vartheta_{1,0}^{k_0}}{k_0!}\frac{\vartheta_{1,1}^{2l-2k_0-3m}}{(2l-2k_0-3m)!}\frac{\vartheta_{1,2}^{k_0-l+2m}}{(k_0-l+2m)!}.
\end{equation}
\end{lemma}
\begin{proof} By the formula of combination, we have
	\begin{equation}
	\begin{split}
	 &\sum_{i=1}^{\infty}\frac{\epsilon^i}{i!}\left(\vartheta_{1,0}+\epsilon\vartheta_{1,1}+\epsilon^2\vartheta_{1,2} \right)^i \\
	=&\sum_{i=1}^{\infty}\sum_{m=0}^{i}\sum_{k_0+k_1+k_2=i}^{ k_1+2k_2=m,\,k_s\geq 0,\,s=1,2,3}\frac{\vartheta_{1,0}^{k_0}}{k_0!}\frac{\vartheta_{1,1}^{k_1}}{k_1!}\frac{\vartheta_{1,2}^{k_2}}{k_2!}\epsilon^{m+i} \\
	 =&\sum_{i=1}^{\infty}\sum_{m=0}^{i}\sum_{k_0=i-m}^{[i-m/2]}\frac{\vartheta_{1,0}^{k_0}}{k_0!}\frac{\vartheta_{1,1}^{2i-m-2k_0}}{(2i-m-2k_0)!}\frac{\vartheta_{1,2}^{k_0+m-i}}{(k_0-i+m)!}\epsilon^{m+i}  \\
	 =&\sum_{l=1}^{\infty}\sum_{m=0}^{[2l/3]}\sum_{k_0=\max(l-2m,0)}^{[l-3m/2]}\frac{\vartheta_{1,0}^{k_0}}{k_0!}\frac{\vartheta_{1,1}^{2l-2k_0-3m}}{(2l-2k_0-3m)!}\frac{\vartheta_{1,2}^{k_0-l+2m}}{(k_0-l+2m)!}\epsilon^{l}
	\end{split}
	\end{equation}
which deduces the equation \eqref{eq:expand-exp}.
\end{proof}
We apply the zero seed solutions to the B\"acklund transformation \eqref{eq:back}. With the aid of Lemma \ref{lemma1}, we obtain the following determinant representation of high-order soliton under single spectral parameter $\lambda_1$:
\begin{equation}\label{eq:high order-soliton}
q^{[N]}=2\frac{\det{\mathbf{G}}}{\det{\mathbf{M}}}
\end{equation}
where
\begin{equation}\label{F1function}
\begin{split}
\mathbf{M}:&=\mathbf{F}_1^{\dag}\mathbf{B}_1\mathbf{F}_1+\mathbf{B}_1, \qquad \mathbf{G}:=\begin{bmatrix}\mathbf{M}&\mathbf{F}_{1,1}^{\dag}\\
{\bf e}_1&0
\end{bmatrix},
\\
{\bf e}_1:&=[1,0,0,\cdots, 0],\qquad \mathbf{B}_1:=\left(\frac{1}{\lambda_1-\lambda_1^*}\binom{i+j-2}{i-1}\left(\frac{\ii}{\lambda_1^*-\lambda_1}\right)^{i+j-2}\right)_{1\leq i,j\leq N},\\
\mathbf{F}_1:&=\ee^{2\theta_1}\begin{bmatrix}1&f_{1,1}&\cdots&f_{1,N-1}\\
0&1&\cdots&f_{1,N-2}\\
\vdots&\vdots&\vdots&\vdots\\
0&0&\cdots&1
\end{bmatrix}=\ee^{2\theta_1}\left(\mathbb{I}+\sum_{i=1}^{N-1}f_{1,i}\mathbf{E}_N^i\right),\qquad \mathbf{E}_N:=\left(\delta_{i,j-1}\right)_{1\leq i,j \leq N},
\end{split}
\end{equation}
and the subscripts $_{1}$ in $\mathbf{F}_1$ and $\mathbf{B}_1$ stand for the matrices depending on the spectral parameter $\lambda_1$, $\mathbf{F}_{1,1}$ represents the first row of matrix $\mathbf{F}_1$, $f_{1,l}$ is given in equation \eqref{eq:f-exp}. To analyze the asymptotic behavior of the high-order single-pole soliton, we rewrite the relative matrices $\mathbf{B}_1$ of the determinant solutions:
\begin{equation}\label{eq:b1-representation}
\mathbf{B}_1=\frac{\mathbf{D}_1\mathbf{P}_N\mathbf{D}_1}{\lambda_1-\lambda_1^*}=\frac{\mathbf{D}_1\mathbf{S}_N^{\dag}\mathbf{S}_N\mathbf{D}_1}{\lambda_1-\lambda_1^*},
\end{equation}
where $\mathbf{P}_N=\left(\binom{i+j-2}{i-1}\right)_{1\leq i,j\leq N}$ stands for the $N$-th order Pascal matrix which can be performed the $LU$ decomposition $\mathbf{P}_N=\mathbf{S}_N^{\dag}\mathbf{S}_N$, and
$$\mathbf{D}_1={\rm diag}\left(1,\left(-\frac{1}{2\lambda_{1I}}\right) ,\cdots, \left(-\frac{1}{2\lambda_{1I}}\right) ^{N-1}\right), \quad\mathbf{S}_N=\left(\binom{j-1}{i-1}\right)_{1\leq i\leq j\leq N}.$$ Due to the Lemma \ref{lemma1}, the leading polynomial of $f_{1,l\geq 1}$ is
$$f_{1,j}=\frac{\vartheta_{1,0}^{j}}{j!}+\text{lower order terms (l.o.t) of $x$ and $t$}.$$

\begin{lemma}\label{lemmaGamma}
\begin{equation}\label{eq:determinant}
\Gamma_j:=\left|\begin{matrix}\frac{1}{j!}&
\frac{1}{(j+1)!}&\cdots&\frac{1}{(N-1)!}\\
\frac{1}{(j-1)!}&\frac{1}{j!}&\cdots&\frac{1}{(N-2)!}\\
\vdots&\vdots&\vdots&\vdots\\
\frac{1}{(2j-N+1)!}&\cdots&\cdots&\frac{1}{j!}
\end{matrix}\right|
=
\left(\frac{0!1!2!\cdots (N-1-j)!}{j!(j+1)!\cdots(N-1)!}\right)
\end{equation}
\end{lemma}
\begin{proof}This determinant can be calculated by subtracting the $i-1$-th row from the $i$-th row subsequently, that is
\begin{equation}\label{eq:determinant}
\begin{split}
\left|\begin{matrix}\frac{1}{j!}&
\frac{1}{(j+1)!}&\cdots&\frac{1}{(N-1)!}\\
\frac{1}{(j-1)!}&\frac{1}{j!}&\cdots&\frac{1}{(N-2)!}\\
\vdots&\vdots&\vdots&\vdots\\
\frac{1}{(2j-N+1)!}&\cdots&\cdots&\frac{1}{j!}
\end{matrix}\right|
=&\left|\begin{matrix}\frac{1}{j!}\frac{N-1-j}{(N-1)}&\frac{1}{(j+1)!}\frac{N-1-(j+1)}{(N-1)}&\cdots&\frac{1}{(N-2)!}&0\\
\frac{1}{(j-1)!}\frac{N-1-j}{(N-2)}&\frac{1}{j!}\frac{N-2-j}{(N-2)}&\cdots&\frac{1}{(N-3)!}&0\\
\vdots&\vdots&\vdots&\vdots&\vdots\\
\frac{1}{(2j-N+2)!}\frac{N-1-j}{j+1}&\frac{1}{(2j-N+3)!}\frac{N-2-j}{j+1}&\cdots&\frac{1}{(j+1)!}&0\\
*&*&*&*&\frac{1}{j!}
\end{matrix}\right|\\
=&\frac{(N-1-j)!}{(N-1)!}\left|\begin{matrix}\frac{1}{j!}&\frac{1}{(j+1)!}&\cdots&\frac{1}{(N-2)!}\\
\frac{1}{(j-1)!}&\frac{1}{j!}&\cdots&\frac{1}{(N-3)!}\\
\vdots&\vdots&\vdots&\vdots\\
\frac{1}{(2j-N+2)!}&\frac{1}{(2j-N+3)!}&\cdots&\frac{1}{j!}
\end{matrix}\right|\\
=&\frac{(N-1-j)!(N-2-j)!}{(N-1)!(N-2)!}\left|\begin{matrix}\frac{1}{j!}&\frac{1}{(j+1)!}&\cdots&\frac{1}{(N-3)!}\\
\frac{1}{(j-1)!}&\frac{1}{j!}&\cdots&\frac{1}{(N-4)!}\\
\vdots&\vdots&\vdots&\vdots\\
\frac{1}{(2j-N+3)!}&\frac{1}{(2j-N+4)!}&\cdots&\frac{1}{j!}
\end{matrix}\right|\\
=&
\left(\frac{0!1!2!\cdots (N-1-j)!}{j!(j+1)!\cdots(N-1)!}\right).
\end{split}
\end{equation}
\end{proof}
Through the expression \eqref{eq:high order-soliton}, we find the leading order term of $\det(\mathbf{M})$ can be expanded as the linear combinations of $\exp(2j(\theta_1+\theta_1^*))$ and $\det(\mathbf{G})$ can be expanded as the linear combinations of $\exp(2j(\theta_1+\theta_1^*)+2\theta_1^*)$. The coefficients of $\exp(2j(\theta_1+\theta_1^*))$ and $\exp(2j(\theta_1+\theta_1^*)+2\theta_1^*)$ are polynomial function with respect to $x$ and $t$. In order to study the asymptotic behavior, we need to give the leading order terms of $\det(\mathbf{M})$ and $\det(\mathbf{G})$.
\begin{lemma}\label{lemma-denom}
The denominator $\det(\mathbf{F}_1^{\dag}\mathbf{B}_1\mathbf{F}_1+\mathbf{B}_1)$ and the numerator $\det{\mathbf{G}}$ of $q^{[N]}$ can be represented as
\begin{equation}\label{eq:det-for-denom}
\begin{split}
\det(\mathbf{F}_1^{\dag}\mathbf{B}_1\mathbf{F}_1+\mathbf{B}_1)&=\sum_{j=0}^{N-1}\ee^{2(N-j)(\theta_1+\theta_1^*)}\mathscr{A}_{j}(x, t)
+\mathscr{A}_0(x, t)+{\rm l.o.t}\\
\det{\mathbf{G}}&=\sum_{j=0}^{N-1}\ee^{2\theta_1^*}\ee^{2(N-j-1)(\theta_1+\theta_1^*)}\mathscr{B}_j(x, t)
+{\rm l.o.t}
\end{split}
\end{equation}
where
\begin{equation}
\begin{split}
\mathscr{A}_j(x, t)&=\left(-\ii\right)^N\left(2\lambda_{1I}\right)^{-(N-j)^2-j^2}\Gamma_j^2|\vartheta_{1, 0}|^{2j(N-j)},\\
\mathscr{B}_j(x, t)&=(-1)^{N-j}(-\ii)^{N-1}(\vartheta_{1,0})^{(N-j-1)(j+1)}(\vartheta^*_{1,0})^{j(N-j)}\left(2\lambda_{1I}\right)^{-(N-j-1)^2-j^2-N+1}
\Gamma_{j}\Gamma_{j+1}.
\end{split}
\end{equation}
\end{lemma}

Based on the leading order term of the numerator and the denominator in lemma \ref{lemma-denom}, we begin to analyze the long-time asymptotics of the soliton.
\begin{lemma}\label{lemmamatch}
The high-order soliton $q^{[N]}$ has a non-zero limit along the following characteristic curves:
$$x=s+v_1t+\frac{N+1-2\kappa}{2\lambda_{1I}}\log(|t|),\qquad (\kappa=1,2,\cdots, N),$$
as $t\to\pm\infty$. Otherwise, it has the vanishing limitation.
\end{lemma}

From Lemma \ref{lemmamatch}, we get a conclusion that if the high-order soliton $q^{[N]}$ has a non-zero limit, then the high-order soliton must move along a specific characteristic curve, and the matching terms must be two adjacent terms. Now we give an example to verify this fact.
\begin{example}\label{eg:q2}
When $N=2, \lambda_1=\ii, a_1^{[0]}=a_1^{[1]}=0$, the second order soliton is
\begin{equation}\label{eq:2ordersoliton}
q^{[2]}=\frac{\left(\mathscr{B}_0(x, t)+\frac{1}{2}\ii\right)\ee^{-6x+24\delta t+4\ii\gamma t+\frac{\ii}{2}\pi}+\left(\mathscr{B}_1(x, t)+\frac{1}{2}\ii\right)\ee^{-2x+8\delta t+4\ii\gamma t+\frac{\ii}{2}\pi}}{\mathscr{A}_{0}(x, t)\ee^{-8x+32\delta t}+\left(\mathscr{A}_{1}(x, t)-\frac{1}{8}\right)\ee^{-4x+16\delta t}+\mathscr{A}_{0}(x, t)},
\end{equation}
where
\begin{equation*}
\begin{split}
\mathscr{B}_0(x, t)&=\ii x-4\gamma t-12\ii\delta t, \quad \mathscr{B}_1(x, t)=-\ii x-4\gamma t+12\ii\delta t,\\
\mathscr{A}_{0}(x, t)&=-\frac{1}{16}, \quad \mathscr{A}_{1}(x, t)=-\left(x^2-24\delta tx+144\delta^2t^2+16\gamma^2t^2\right).
\end{split}
\end{equation*}
As given in the lemma \ref{lemmamatch}, taking $N=2, \kappa=1$ and $N=2, \kappa=2$, then the characteristic curves are
\begin{equation}\label{eq:curve1}
\begin{split}
x&=s+4\delta t-\frac{1}{2}\log(|t|),\\
x&=s+4\delta t+\frac{1}{2}\log(|t|).
\end{split}
\end{equation}
Along the first curve, $q^{[2]}$ becomes
\begin{equation}
\begin{split}
q^{[2]}_{\pm}&=\frac{\left[\pm64\left(\ii\gamma-2\delta\right)t^4+O(t^3\log(t))\right]\ee^{-6s+4\ii\gamma t}\pm\left[64\left(\ii\gamma+2\delta\right)t^2+O(t\log(t))\right]\ee^{-2s+4\ii\gamma t}}{t^4\ee^{-8s}+\left[256\left(\gamma^2+4\delta^2\right)t^4+O(t^3\log(t))\right]\ee^{-4s}+1}\\
&=2\ee^{4\ii\gamma t\pm\ii\arctan\left(\frac{-\gamma}{2\delta}\right)+\ii\pi}\sech\left(2\left(x-4\delta t\pm\frac{1}{2}\log(|t|)\right)\pm4\log(2)\pm\frac{1}{2}\log(4\delta^2+\gamma^2)\right)+O(\log(t)/t).
\end{split}
\end{equation}
Along the second curve, $q^{[2]}$ becomes
\begin{equation}
\begin{split}
q^{[2]}_{\pm}&=\frac{\left[\pm64\left(\ii\gamma-2\delta\right)t^{-2}+O(t^{-3}\log(t))\right]\ee^{-6s+4\ii\gamma t}\pm\left[64\left(\ii\gamma+2\delta\right)+O(t^{-1}\log(t))\right]\ee^{-2s+4\ii\gamma t}}{t^{-4}\ee^{-8s}+\left[256\left(\gamma^2+4\delta^2\right)+O(t^{-1}\log(t))\right]\ee^{-4s}+1}\\
&=2\ee^{4\ii\gamma t\pm\ii\arctan\left(\frac{\gamma}{2\delta}\right)}\sech\left(2\left(x-4\delta t\mp\frac{1}{2}\log(|t|)\right)\mp4\log(2)\mp\frac{1}{2}\log(4\delta^2+\gamma^2)\right)+O(\log(t)/t).
\end{split}
\end{equation}
Therefore, if the high-order soliton moves along these two characteristic curves, the long-time asymptotics of $q^{[2]}$ is
\begin{equation}\label{q2:asym}
\begin{split}
q^{[2]}_{\pm}&=2\ee^{4\ii\gamma t\pm\ii\arctan\left(\frac{-\gamma}{2\delta}\right)+\ii\pi}\sech\left(2\left(x-4\delta t\pm\frac{1}{2}\log(|t|)\right)\pm4\log(2)\pm\frac{1}{2}\log(4\delta^2+\gamma^2)\right)\\
&+2\ee^{4\ii\gamma t+\ii\arctan\left(\frac{\gamma}{2\delta}\right)}\sech\left(2\left(x-4\delta t\mp\frac{1}{2}\log(|t|)\right)\mp4\log(2)\mp\frac{1}{2}\log(4\delta^2+\gamma^2)\right)+O(\log(t)/t).
\end{split}
\end{equation}
\end{example}
\ys{In \cite{Cen-PhysD-2019}, the authors utilized two different methods to get the degenerate multi-soliton and analyzed their asymptotics. By choosing a special parameter $\lambda_1$ in present work and the corresponding parameter $\mu$ in \cite{Cen-PhysD-2019}, such as, we set $\lambda_1=\frac{\ii}{2}\mu, x\to x-\frac{\log(4\lambda_{1I})}{2\lambda_{1I}},$ then we have $ q^{[1]}\ee^{-\ii \frac{\lambda_{1R}\log{(4\lambda_{1I})}}{\lambda_{1I}}}=q_{1}^{\mu}$, where $q^{[1]}$ is taken by $N=1$ in Eq.\eqref{eq:high order-soliton} and $q_1^{\mu}$ is Eq.(3.7) in \cite{Cen-PhysD-2019}. Under this parameter setting, the asymptotics of the second order solitons in Example \ref{eg:q2} is consistent with the Eq.(4.8) in \cite{Cen-PhysD-2019}. In \cite{Xu-JPSJ-2020}, the authors gave the second order soliton to this Hirota equation and analyzed its asymptotics when $t\to\infty$. Our result in Example \ref{eg:q2} is also consistent with the second order soliton in Eq. (19) by choosing $\chi=1-\frac{1}{12\epsilon}\ii, \gamma_1=1, s_1=0$ in \cite{Xu-JPSJ-2020} and $\gamma=\frac{1}{2}, \delta=\varepsilon$ in Example \ref{eg:q2}. Correspondingly, they have the same asymptotic expression when $t\to\pm\infty$, which can be seen from Eq. \eqref{q2:asym} in our work and Eq.(26) in \cite{Xu-JPSJ-2020}. Moreover, the asymptotics for the third order soliton is also compatible.}
With the Lemma \ref{lemmamatch} and the Example \ref{eg:q2}, we give a theorem about the asymptotic expression of the $N$-th order soliton with single spectral parameter $\lambda_1$.
\begin{theorem}\label{theorem1}
As $t\to\pm\infty$, the long-time asymptotics of the high-order soliton $q^{[N]}$ can be represented as
\begin{equation}
\begin{split}
q^{[N]}_{\pm}&{=}\sum_{\kappa{=}1}^{\left[\frac{N+1}{2}\right]}2\lambda_{1I} \sech\left(2\lambda_{1I}s_{1, l,\pm}{\pm}\Delta_{1,l}^{[\kappa]}-2a_1^{[0]}\right)\ee^{-2\ii{\rm Im}(\theta_1)\pm\ii\frac{N+1-2\kappa}{2}\left(\arg\left(\frac{\mathscr{C}_1}{\mathscr{C}_1^*}\right)\right)+\ii\pi\left(\frac{1}{2}+\kappa \right)}
\\&{+}\sum_{\kappa=1}^{\left[\frac{N}{2}\right]}2\lambda_{1I}
 \sech\left(2\lambda_{1I}s_{1, r,\pm}\mp\Delta_{1, r}^{[\kappa]}-2a_1^{[0]}\right)\ee^{-2\ii{\rm Im}(\theta_1)
 \mp\ii\frac{N+1-2\kappa}{2}\left(\arg\left(\frac{\mathscr{C}_1}{\mathscr{C}_1^*}\right)\right)+\ii\pi\left(N-\frac{1}{2}+\kappa\right)}+O(\log(t)/t),
\end{split}
\end{equation}
\normalsize
where the subscript $_\pm$ stands for $t\to\pm\infty$ and
\begin{equation*}
\begin{split}
s_{i, l, \pm}&=x\mp\left(\frac{N+1-2\kappa}{2\lambda_{iI}}\right)\log(|t|)-v_{i}t,\\\
s_{i, r, \pm}&=x\pm\left(\frac{N+1-2\kappa}{2\lambda_{iI}}\right)\log(|t|)-v_it,\\
\Delta_{1, l}^{[\kappa]}&=\log\left(\frac{(N-\kappa)!}{(\kappa-1)!}\right)-(N+1-2\kappa)
\left(\log\left(2\lambda_{1I}\right)+\log\left(|\mathscr{C}_1|\right)\right),\\
\Delta_{1, r}^{[\kappa]}&=\log\left(\frac{(N-\kappa)!}{(\kappa-1)!}\right) {-}(N+1-2k)\left(\log\left(2\lambda_{1I}\right){+}
\log\left(|\mathscr{C}_1|\right)\right),\\
\mathscr{C}_i&=16\delta\lambda_{iI}^2+8\ii\gamma\lambda_{iI}+48\ii\delta\lambda_{iR}\lambda_{iI}.
\end{split}
\end{equation*}
and the subscript $_{l,r}$ means the the left characteristic curves or the right one.
\end{theorem}
\begin{proof}
With Lemma \ref{lemmamatch}, if the high-order soliton moves along a specific curve
\begin{equation}\label{eq:curve3}
x=s_{1, l,+}+\left(\frac{N+1-2\kappa}{2\lambda_{1I}}\right)\log(t)
+v_1t,
\end{equation}
for every $\kappa=1,2,\cdots, \left[\frac{N+1}{2}\right],$ then the high-order soliton $q^{[N]}$ becomes
\small
\begin{equation}
\begin{split}
q^{[N]}&=\frac{(-1)^{\kappa} 2\ii\lambda_{1I}\mathscr{C}_1^{N+1-2\kappa}\ee^{2\lambda_{1I}s_{1,l,+}-2a_{1}^{[0]}-2\ii{\rm Im}(\theta_1)}
t^{2\kappa(\kappa-1)}+O(t^{2\kappa(\kappa-1)-1}\log(t))}{\left[\ee^{4\lambda_{1I}s_{1, l,+}-4a_{1}^{[0]}}
(2\lambda_{1I})^{2\kappa-N-1}\left(\frac{(N-\kappa)!}{(\kappa-1)!}\right)
+|\mathscr{C}_1|^{2N+2-4\kappa}(2\lambda_{1I})^{N+1-2\kappa}
\left(\frac{(\kappa-1)!}{(N-\kappa)!}\right)\right]t^{2\kappa(\kappa-1)}+O(t^{2\kappa(\kappa-1)-1}\log(t))}\\
&=2\lambda_{1I} \sech\left(2\lambda_{1I}s_{1, l,+}{+}\Delta_{1,l}^{[\kappa]}-2a_{1}^{[0]}\right)\ee^{-2\ii{\rm Im}(\theta_1)+\ii\frac{N+1-2\kappa}{2}
\left(\arg\left(\frac{\mathscr{C}_1}{\mathscr{C}_1^*}\right)\right)+\ii\pi\left(\frac{1}{2}+\kappa \right)}+O(\log(t)/t).
\end{split}
\end{equation}
\normalsize
Similarly, along the other characteristic curve
$$x=s_{1, r,+}-\left(\frac{N+1-2\kappa}{2\lambda_{1I}}\right)\log(t)
+v_1t,\qquad \left(\kappa=1,2,\cdots, \left[\frac{N}{2}\right]\right),$$
the asymptotic expression is
\small
\begin{equation}
\begin{split}
q^{[N]}&=\frac{\ii(-1)^{N-1+\kappa}2\lambda_{1I}(\mathscr{C}_1^*)^{N+1-2\kappa}
\ee^{2\lambda_{1I}s_{1, r,+}-2a_{1}^{[0]}-2\ii{\rm Im}(\theta_1)}t^{2(\kappa-N)(\kappa-N-1)}+O(t^{2(\kappa-N)(\kappa-N-1)-1}\log(t))}
{\left[\frac{(\kappa-1)!}{(N-\kappa)!}\ee^{4\lambda_{1I}s_{1, r,+}-4a_{1}^{[0]}}|\mathscr{C}_1|^{2N+2-4\kappa}(2\lambda_{1I})^{N+1-2\kappa}
+\frac{(N-\kappa)!}{(\kappa-1)!}(2\lambda_{1I})^{2\kappa-N-1}\right]t^{2(\kappa-N)(\kappa-N-1)}+O(t^{2(\kappa-N)(\kappa-N-1)-1}\log(t))}\\
&=2\lambda_{1I}
\sech\left(2\lambda_{1I}s_{1, r,+}-\Delta_{1, r}^{[\kappa]}-2a_{1}^{[0]}\right)\ee^{-2\ii{\rm Im}(\theta_1)-\ii\frac{N+1-2\kappa}{2}\left(\arg\left(\frac{\mathscr{C}_1}{\mathscr{C}_1^*}\right)\right)
+\ii\pi\left(N-\frac{1}{2}+\kappa\right)}+O(\log(t)/t).
\end{split}
\end{equation}
\normalsize
If the high-order soliton moves along the other curves $x=s+vt+\beta\log(t)$ with $v\neq v_1$ or $\beta\neq\frac{N+1-2\kappa}{2\lambda_{1I}}$, then its asymptotic solution is
\begin{equation}
q^{[N]}=O(\ee^{-2\lambda_{1I}|v-v_1|t}),
\end{equation}
and
\begin{equation}
q^{[N]}=O(t^{-1}),
\end{equation}
respectively, both of which are vanishing when $t\to\infty$.
Hence, the global long-time asymptotics of $q^{[N]}$ as $t\to+\infty$ is
\begin{equation}
\begin{split}
q^{[N]}&{=}\sum_{\kappa{=}1}^{\left[\frac{N+1}{2}\right]}2\lambda_{1I} \sech\left(2\lambda_{1I}s_{1, l,+}{+}\Delta_{1,l}^{[\kappa]}{-}2a_{0}^{[1]}\right)\ee^{-2\ii{\rm Im}(\theta_1)+\ii\frac{N+1-2\kappa}{2}\left(\arg\left(\frac{\mathscr{C}_1}{\mathscr{C}_1^*}\right)\right)+\ii\pi\left(\frac{1}{2}+\kappa \right)}
\\&{+}\sum_{\kappa=1}^{\left[\frac{N}{2}\right]}2\lambda_{1I}
 \sech\left(2\lambda_{1I}s_{1, r,+}{-}\Delta_{1, r}^{[\kappa]}{-}2a_{0}^{[1]}\right)\ee^{-2\ii{\rm Im}(\theta_1)
 -\ii\frac{N+1-2\kappa}{2}\left(\arg\left(\frac{\mathscr{C}_1}{\mathscr{C}_1^*}\right)\right)+\ii\pi\left(N-\frac{1}{2}+\kappa \right)}+O(\log(t)/t).
\end{split}
\end{equation}
When $t\to-\infty$, its asymptotic expression can be given in a similar manner.
\end{proof}

\ys{\begin{remark}
From the asymptotic expression in above theorem, we can see the position and phase-shifts clearly. When $t\to \pm\infty$, we first give the time-dependent displacement between two adjacent solitons about the left characteristic curves, the right ones can be analyzed similarly. For every $\kappa=1,2 \cdots, \left[\frac{N+1}{2}\right]$, denote
\begin{equation}
\begin{split}
x_{\pm}^{[\kappa]}&=v_1t\mp\left(\frac{1}{2\lambda_{1I}}\log\left(\frac{(N-\kappa)!}{(\kappa-1)!}\right)-\frac{1}{2\lambda_{1I}}(N+1-2\kappa)\log(2\lambda_{1I})\right),\\
\Delta_{t}^{[\kappa]}&=\frac{N+1-2\kappa}{2\lambda_{1I}}\log(\left|\mathscr{C}_1t\right|),
\end{split}
\end{equation}
when $t\to\pm\infty$, we have
\begin{equation}
\lim\limits_{t\to\pm\infty}|q^{[N]}(x_{\pm}^{[\kappa]}\pm\Delta_{t}^{[\kappa]},t)|=2\lambda_{1I}.
\end{equation}
Then for two adjacent two solitons, the time-dependent displacement is
\begin{equation}\label{eq:time-dis}
\Delta^{[\kappa]}=\frac{1}{\lambda_{1I}}\left(\log(|\mathscr{C}_1t|)+\log(2\lambda_{1I})\right)-\frac{1}{2\lambda_{1I}}\log\left(\kappa(N-\kappa)\right),
\end{equation}
which is related to the parameters $\kappa, N, \lambda_{1I}$ and $\mathscr{C}_1$. Clearly, the time-dependent displacement grows logarithmically with $t$. Moreover, for fixed $N, \lambda_{1I}$ and $\mathscr{C}_1$, when $\kappa=\left[\frac{N+1}{2}\right]$, $\Delta^{[\kappa]}$ attains the minimum. That is, the time-displacement between the middle of two solitons is the smallest.  With the above formula \eqref{eq:time-dis}, we can give the time-dependent displacement to arbitrary two asymptotic solitons when $t\to\pm\infty$. Specially, when $\kappa=1, N=2$, Eq.\eqref{eq:time-dis} agrees with the time-dependent displacement $2\delta \ln(2\delta)+2\Delta(t)$ in \cite{Cen-PhysD-2019} by choosing $\lambda_1=\frac{\ii}{2}\mu$. By the way, we can also give the maximum time-dependent displacement as
\begin{equation}
\Delta_{\max}=\frac{N-1}{\lambda_{1I}}\left(\log |\mathscr{C}_1t|+\log\left(2\lambda_{1I}\right)\right)-\frac{\log(N-1)!}{\lambda_{1I}}.
\end{equation}

\end{remark}}
\lm{Through this theorem, we give the asymptotic expression with one spectral parameter $\lambda_1$ completely. More significant, our method can also be applied to analyze the asymptotic behavior for other integrable equation. Especially, to the AKNS hierarchy, we just need a slight change to the velocity $v_i$ in theorem \ref{theorem1} to obtain the similar asymptotic expression under the framework of Darboux transformation. Additionally, this method can be extended to analyze the asymptotics to the derivative NLS equation or other integrable equation. In a word, this method provides a tool for analyzing the asymptotics for high-order soliton.} Next, we will give several examples to verify the above asymptotic expressions by numerical plotting.
\begin{example}
Firstly, we give the even-th order figures in Fig.\ref{fig:even-order} by choosing $\lambda_1=\ii, \delta=1, \gamma=1/2, a_1^{[j]}=0(j=0,1,\cdots, 7)$. With the simple calculation, we know the velocity is $v_1=4.$ If we use a simple transformation $\xi=x-4t$, then the velocity is zero on the $(\xi, t)$-plane.
\begin{figure}[!h]
{\includegraphics[height=0.45\textwidth]{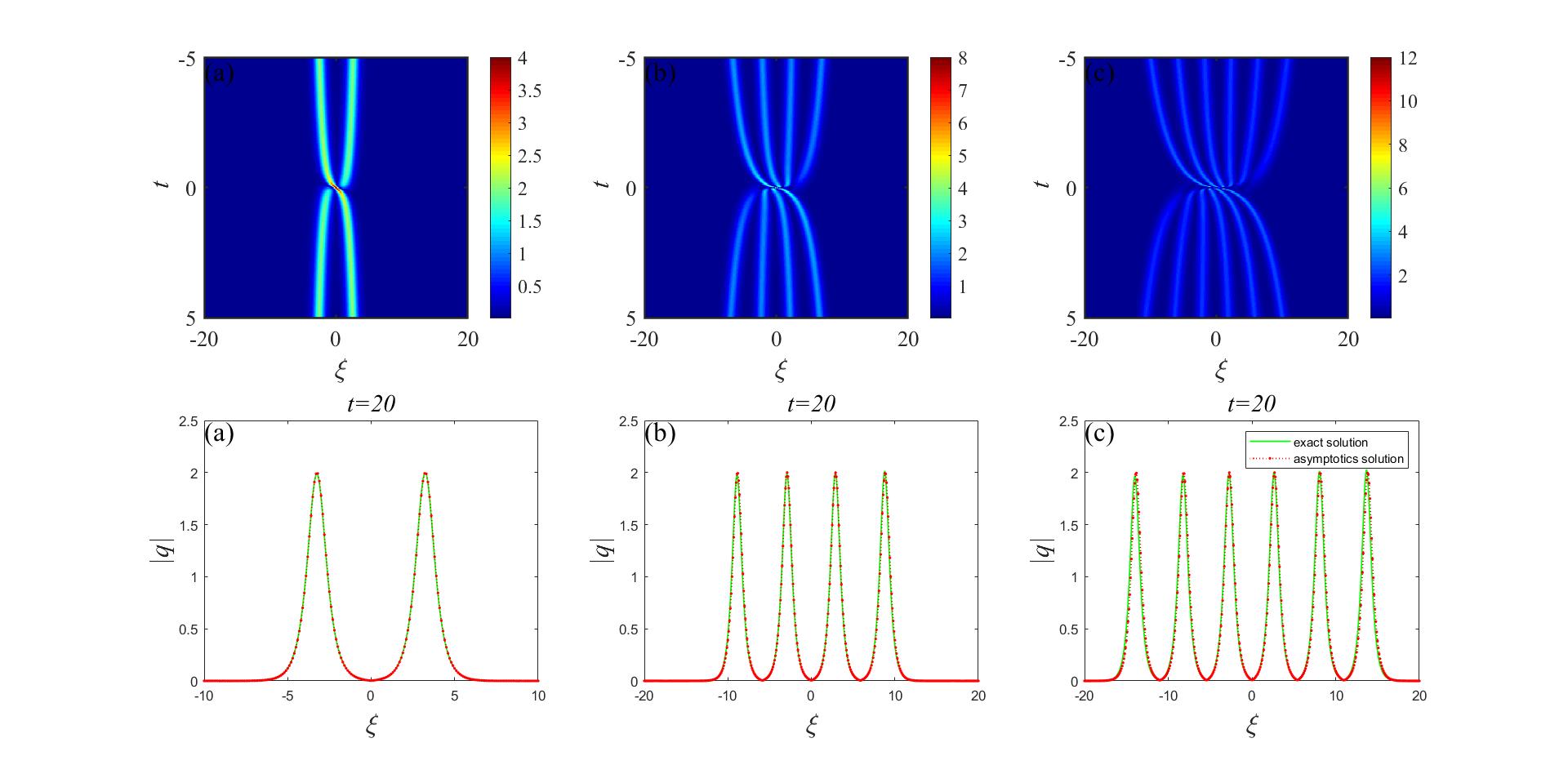}}
\caption{\small The upper is the even-th order soliton, (a) is second-order soliton, (b) is fourth-order soliton, (c) is sixth-order soliton. The below is the comparison between the exact solution and the numerical result.}\label{fig:even-order}
\end{figure}
Secondly, we give the odd-th order figures in Fig.\ref{fig:odd-order} by choosing the same parameters:
\begin{figure}[!h]
{\includegraphics[height=0.45\textwidth]{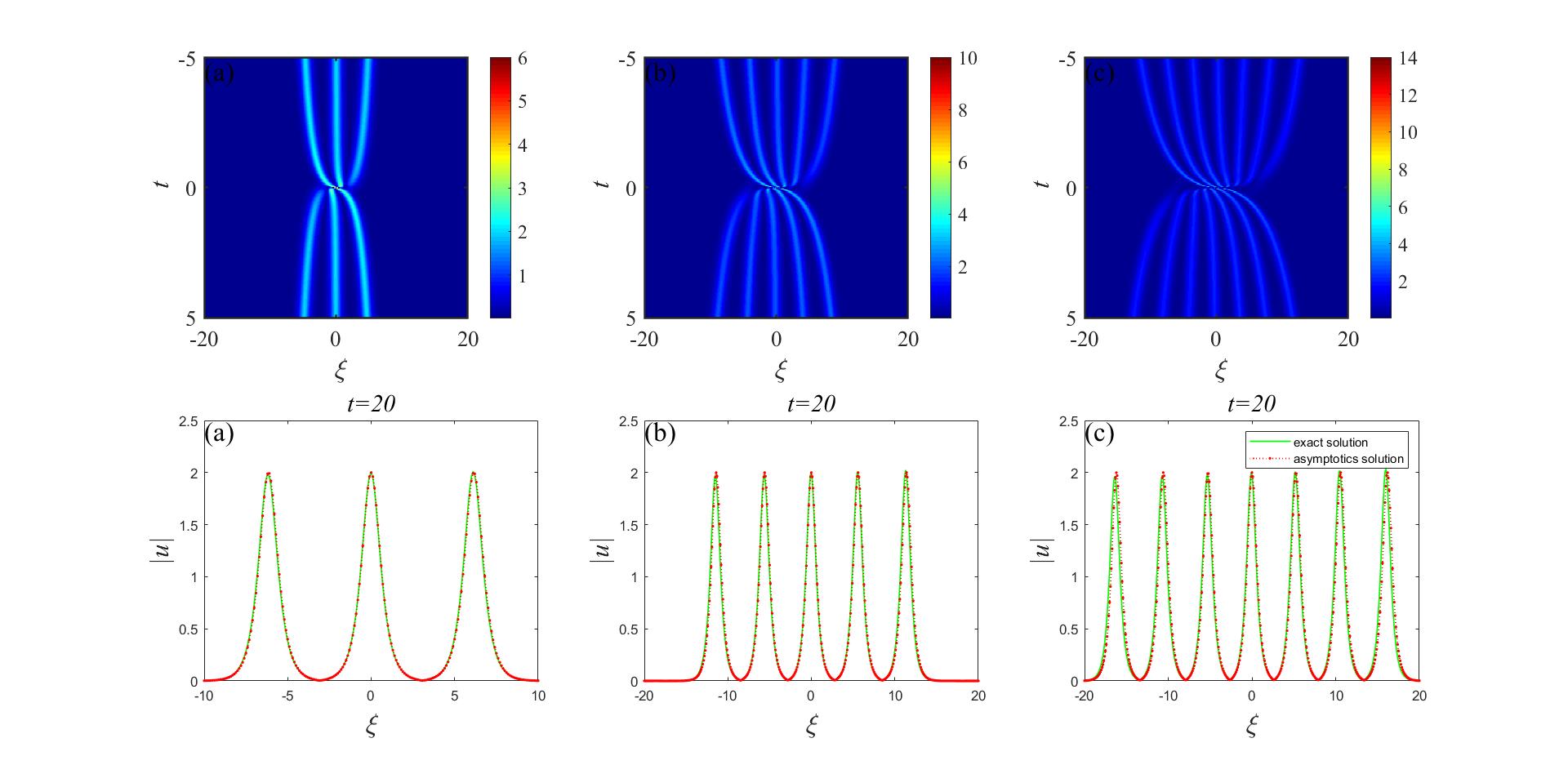}}
\caption{\small  The upper is the odd-th order soliton, (a) is third-order soliton, (b) is fifth-order soliton, (c) is seventh-order soliton. The below is the comparison between the exact solution and the numerical result.}\label{fig:odd-order}
\end{figure}
\end{example}
In the above figures, we choose $a_1^{[j]}=0(j=0,1,\cdots, 7)$ to attain the maximal value at $(x,t)=(0,0)$. Actually, the peak of $N$-order soliton can be calculated with the B\"{a}cklund transformation \eqref{eq:back}.
\begin{prop}
When $x=0, t=0, a_1^{[j]}=0,$ $(j=0,1,\cdots, N-1)$, the high-order solution $q^{[N]}$ has the maximum amplitude $|q^{[N]}(0,0)|=2|\lambda_{1I}|N$.
\end{prop}
\begin{proof}
When $x=0, t=0,$ $F_{1}=\mathbb{I}$, so the determinant $\det(\mathbf{M})=\det(2\mathbf{B}_1)$. With the decomposition Eq.\eqref{eq:b1-representation}, it is obvious that
\begin{equation}
\det(2\mathbf{B}_1)=\left(-\ii\right)^{N}2^{N(1-N)}\left(\lambda_{1I}\right)^{-N^2}.
\end{equation}
Moreover, with the Laplace expansion, we have
\begin{equation}
\begin{split}
\det{\mathbf{G}}&=(-\ii)^{N}\left(\lambda_{1I}\right)^{1-N^2}2^{N(1-N)}N,
\end{split}
\end{equation}
so
\begin{equation*}
\left|q^{[N]}\right|(0,0)=2|\lambda_{1I}|N.
\end{equation*}
\end{proof}
\subsection{The general asymptotics with $k$ spectral parameters}
After giving the long-time asymptotics to the high-order soliton with single spectral parameter $\lambda_1$, we intend to study the asymptotics to the high-order soliton with multiple spectral parameters. For the simplicity, we first discuss the asymptotic expression with two spectral parameters. To avoid the tedious analysis on the determinants, we would like to utilize the iterated Darboux matrix to analyze the dynamics. Then we give the following lemma
\begin{lemma}\label{lemma2spectral}
Suppose there are two spectral parameters $\lambda_1$ and $\lambda_2$, we can rewrite the Darboux matrix in Eq.\eqref{eq:darboux} to the following form
\begin{equation}\label{eq:darboux12}
\mathbf{T}_{N}(\lambda; x, t)=\mathbf{T}_1^{[n_1]}\mathbf{T}_{1}^{[n_1-1]}\cdots\mathbf{T}_{1}^{[1]}
\mathbf{T}^{[n_2]}_2\mathbf{T}_{2}^{[n_2-1]}\cdots\mathbf{T}_{2}^{[1]},
\end{equation}
or
\begin{equation}\label{eq:darboux21}
\mathbf{T}_{N}(\lambda; x, t)=\widetilde{\mathbf{T}}_2^{[n_2]}\widetilde{\mathbf{T}}_{2}^{[n_2-1]}\cdots\widetilde{\mathbf{T}}_{2}^{[1]}
\widetilde{\mathbf{T}}^{[n_1]}_1\widetilde{\mathbf{T}}_{1}^{[n_1-1]}\cdots\widetilde{\mathbf{T}}_{1}^{[1]},
\end{equation}
where
\begin{equation}\label{eq:high-dt}
\begin{split}
\mathbf{T}_{i}^{[j]}&=\left(\mathbb{I}-\frac{\lambda_i-\lambda_i^*}{\lambda-\lambda_i^*}\mathbf{P}^{[j]}_i\right),\quad\left(i=1,2;\,\, j=1,\cdots, n_i\right),\\
\mathbf{P}_{i}^{[j]}&=\frac{\Phi_i^{[j-1]}\left(\Phi_{i}^{[j-1]}\right)^{\dagger}}{\left(\Phi_{i}^{[j-1]}\right)^{\dagger}\Phi_i^{[j-1]}}, \quad \mathbf{\Phi}_2^{[j]}=\lim\limits_{\epsilon\to 0}\frac{\left(\mathbf{T}_{2}^{[j]}\mathbf{T}_2^{[j-1]}\cdots
\mathbf{T}_{2}^{[0]}\right)\Phi(\lambda)}{\epsilon^j}\Bigg|_{\lambda=\lambda_2+\epsilon},\qquad \mathbf{T}_i^{[0]}=\mathbb{I},\\
\mathbf{\Phi}_1^{[j]}&=\lim\limits_{\epsilon\to 0}\frac{\left(\mathbf{T}_{1}^{[j]}\mathbf{T}_1^{[j-1]}\cdots\mathbf{T}_{1}^{[0]}
\mathbf{T}_{2}^{[n_2]}\mathbf{T}_2^{[n_2-1]}\cdots\mathbf{T}_{2}^{[0]}\right)\Phi(\lambda)}{\epsilon^j}\Bigg|_{\lambda=\lambda_1+\epsilon},\\
\widetilde{\mathbf{T}}_{i}^{[j]}&=\left(\mathbb{I}-\frac{\lambda_i-\lambda_i^*}{\lambda-\lambda_i^*}\widetilde{\mathbf{P}}^{[j]}_i\right),\quad\left(i=1,2;\,\, j=1,\cdots, n_i\right),\\
\widetilde{\mathbf{P}}_{i}^{[j]}&=\frac{\widetilde{\Phi}_i^{[j-1]}\left(\widetilde{\Phi}_{i}^{[j-1]}\right)^{\dagger}}{\left(\widetilde{\Phi}_{i}^{[j-1]}\right)^{\dagger}\widetilde{\Phi}_i^{[j-1]}}, \quad\widetilde{ \mathbf{\Phi}}_1^{[j]}=\lim\limits_{\epsilon\to 0}\frac{\left(\widetilde{\mathbf{T}}_{1}^{[j]}\widetilde{\mathbf{T}}_1^{[j-1]}\cdots\widetilde{\mathbf{T}}_{1}^{[0]}\right)\Phi(\lambda)}{\epsilon^j}\Bigg|_{\lambda=\lambda_1+\epsilon},\qquad \widetilde{\mathbf{T}}_{i}^{[0]}=\mathbb{I},\\
\widetilde{\mathbf{\Phi}}_2^{[j]}&=\lim\limits_{\epsilon\to 0}\frac{\left(\widetilde{\mathbf{T}}_{2}^{[j]}\widetilde{\mathbf{T}}_2^{[j-1]}\cdots\widetilde{\mathbf{T}}_{2}^{[0]}\widetilde{\mathbf{T}}_{1}^{[n_1]}
\widetilde{\mathbf{T}}_2^{[n_1-1]}\cdots\widetilde{\mathbf{T}}_{1}^{[0]}\right)\Phi(\lambda)}{\epsilon^j}\Bigg|_{\lambda=\lambda_2+\epsilon},\\
\end{split}
\end{equation}
where $\Phi(\lambda)$ belongs to the one dimension linear space: $\mathrm{span}\{\text{vector solution for Lax pair}\}$.
If the soliton moves along the trajectory $\xi=x-vt$ with $v<v_2$, the asymptotics of the Darboux matrix Eq.\eqref{eq:darboux12} can be represented as
\begin{equation}\label{eq:dtasym}
\mathbf{T}^{[n_2]}_2\mathbf{T}_{2}^{[n_2-1]}\cdots\mathbf{T}_{2}^{[1]}=\left\{\begin{split} &\begin{bmatrix}
\left(\frac{\lambda-\lambda_2}{\lambda-\lambda_2^*}\right)^{n_2}&0\\0&1
\end{bmatrix}+O\left(t^{n_2-1}\ee^{-2\lambda_{2I}|v-v_2|t}\right),\quad  t\to +\infty,\\
&\begin{bmatrix}
1&0\\0&\left(\frac{\lambda-\lambda_2}{\lambda-\lambda_2^*}\right)^{n_2}
\end{bmatrix}+O\left(|t|^{n_2-1}\ee^{2\lambda_{2I}|v-v_2|t}\right), \quad t\to -\infty.
 \end{split}\right.
\end{equation}
Similarly, if the soliton moves along the trajectory $\xi=x-vt$ with $v>v_1$, then the asymptotics of the Darboux matrix in Eq.\eqref{eq:darboux21} is
\begin{equation}\label{eq:dtasym1}
\widetilde{\mathbf{T}}^{[n_1]}_1\widetilde{\mathbf{T}}_{1}^{[n_1-1]}\cdots\widetilde{\mathbf{T}}_{1}^{[1]}=\left\{\begin{split} &\begin{bmatrix}
1&0\\0&\left(\frac{\lambda-\lambda_1}{\lambda-\lambda_1^*}\right)^{n_1}
\end{bmatrix}+O\left(t^{n_1-1}\ee^{-2\lambda_{1I}|v-v_1|t}\right),\quad  t\to +\infty,\\
&\begin{bmatrix}
\left(\frac{\lambda-\lambda_1}{\lambda-\lambda_1^*}\right)^{n_1}&0\\0&1
\end{bmatrix}+O\left(|t|^{n_1-1}\ee^{2\lambda_{1I}|v-v_1|t}\right), \quad t\to -\infty.
 \end{split}\right.
\end{equation}
\end{lemma}
\begin{proof}
We just need to consider one peculiar case. The other cases can be considered similarly. Based on the idea\cite{Faddeev-Book}, when $t\to+\infty$, if the soliton moves along the trajectory $\xi=x-vt$ with $v<v_2, \lambda_{2I}>0$, following Eq.\eqref{eq:high-dt}, we have
\begin{equation}\label{eq:phi20}
\begin{split}
\mathbf{\Phi}_2^{[0]}&=\begin{bmatrix}1\\\ee^{-2\theta_2}
\end{bmatrix}=\begin{bmatrix}1\\0
\end{bmatrix}+O(\ee^{-2\lambda_{2I}|v-v_2|t}),\quad t\to+\infty,\\
\end{split}
\end{equation}
where $\theta_2:=\theta^{[2]}\big|_{\lambda=\lambda_2}=\ii\lambda_2(x+(2\gamma\lambda_2+4\delta\lambda_2^2)t)
+a_2^{[0]}-\frac{\ii}{4}\pi, \lambda_2=\lambda_{2R}+\ii\lambda_{2I}.$ Correspondingly, the asymptotics for $\mathbf{P}_{2}^{[1]}$ and $\mathbf{T}_{2}^{[1]}$ can be represented as
\begin{equation}\label{eq:p20t20}
\begin{split}
\mathbf{P}_2^{[1]}&=\frac{\mathbf{\Phi}_{2}^{[0]}\left(\mathbf{\Phi}_{2}^{[0]}\right)^{\dagger}}
{\left(\mathbf{\Phi}^{[0]}\right)^{\dagger}\mathbf{\Phi}^{[0]}}=\begin{bmatrix}1&0\\0&0
\end{bmatrix}+O(\ee^{-2\lambda_{2I}|v-v_2|t}),\quad t\to+\infty,\\ \mathbf{T}_{2}^{[1]}&=\left(\mathbb{I}-\frac{\lambda_2-\lambda_2^*}{\lambda-\lambda_2^*}\mathbf{P}_{2}^{[1]}\right)=\begin{bmatrix}\left(\frac{\lambda-\lambda_2}{\lambda-\lambda_2^*}\right)&0\\0&1
\end{bmatrix}+O(\ee^{-2\lambda_{2I}|v-v_2|t}), \quad t\to+\infty.\\
\end{split}
\end{equation}
Furthermore, with Eq.\eqref{eq:high-dt}, Eq.\eqref{eq:phi20} and Eq.\eqref{eq:p20t20}, we have
\begin{equation}
\begin{split}
\mathbf{\Phi}_2^{[1]}&=\left(\mathbb{I}-\mathbf{P}_2^{[1]}\right)\begin{bmatrix}0\\ \frac{d}{d\lambda}\ee^{-2\theta^{[2]}}\Big|_{\lambda=\lambda_2}\end{bmatrix}
+\frac{\mathbf{P}_{2}^{[1]}}{\lambda_2-\lambda_2^*}\begin{bmatrix}1\\\ee^{-2\theta_2}\end{bmatrix}
=\frac{1}{\lambda_2-\lambda_2^*}\begin{bmatrix}1\\0
\end{bmatrix}+O\left(t\ee^{-2\lambda_{2I}|v-v_2|t}\right),\, t\to+\infty,\\
\mathbf{P}_2^{[2]}&=\frac{\mathbf{\Phi}_{2}^{[1]}\left(\mathbf{\Phi}_{2}^{[1]}\right)^{\dagger}}
{\left(\mathbf{\Phi}_{2}^{[1]}\right)^{\dagger}\mathbf{\Phi}_{2}^{[1]}}=\begin{bmatrix}1&0\\0&0
\end{bmatrix}+O\left[t\ee^{-2\lambda_{2I}|v-v_2|t}\right],\quad t\to+\infty,\\ \mathbf{T}_{2}^{[2]}&=\left(\mathbb{I}-\frac{\lambda_2-\lambda_2^*}{\lambda-\lambda_2^*}
\mathbf{P}_{2}^{[2]}\right)=\begin{bmatrix}\left(\frac{\lambda-\lambda_2}{\lambda-\lambda_2^*}\right)&0\\0&1
\end{bmatrix}+O\left[t\ee^{-2\lambda_{2I}|v-v_2|t}\right], \quad t\to+\infty.\\
\end{split}
\end{equation}
In succession, after $n_2$-fold iteration, we have
\begin{equation}
\begin{split}
\mathbf{T}_{2}^{[n_2]}\mathbf{T}_{2}^{[n_2-1]}\cdots\mathbf{T}_{2}^{[1]}=\begin{bmatrix}\left(\frac{\lambda-\lambda_2}{\lambda-\lambda_2^*}\right)^{n_2}&0\\0&1
\end{bmatrix}+O\left(t^{n_2-1}\ee^{-2\lambda_{2I}|v-v_2|t}\right), \quad t\to+\infty,
\end{split}
\end{equation}
which deduces the equation in \eqref{eq:dtasym}.

If the soliton moves along the trajectory $\xi=x-vt$ with $v>v_1$, in this case, we use the Darboux matrix in Eq.\eqref{eq:darboux21}, then its asymptotics can be given in a similar manner. Similarly, the asymptotics of the Darboux matrix when $t\to-\infty$ can also be presented. It completes the proof.
\end{proof}
Based on the theorem \ref{theorem1} and the asymptotics of Darboux matrix in Eq.\eqref{eq:dtasym} and Eq.\eqref{eq:dtasym1}, the high-order soliton $q^{[N]}$ can be expressed as a sum of $n_1$ single soltions along the specific characteristic curve with the vector $\mathbf{\Phi}_{1}^{[0]}$ defined as
\begin{equation}\label{eq:phi1-asymptotics}
\begin{split}
\mathbf{\Phi}_{1}^{[0]}=&
\left\{\begin{split}&\begin{bmatrix}\left(\frac{\lambda_1-\lambda_2}{\lambda_1-\lambda_2^*}\right)^{n_2}
\ee^{2\theta_1}\\1
\end{bmatrix},\qquad t\to +\infty,\\
&\begin{bmatrix}\left(\frac{\lambda_1-\lambda_2^*}{\lambda_1-\lambda_2}\right)^{ n_2}
\ee^{2\theta_1}\\1
\end{bmatrix},\qquad t\to -\infty. \end{split}\right.
\end{split}
\end{equation}
Additionally, $q^{[N]}$ can also be expressed as a sum of $n_2$ single solitons along another specific characteristic curve with the vector $\mathbf{\Phi}_{2}^{[0]}$ defined as
\begin{equation}\label{eq:phi2-asymptotics}
\begin{split}
\mathbf{\Phi}_{2}^{[0]}=&
\left\{\begin{split}&\begin{bmatrix}\left(\frac{\lambda_2-\lambda_1}{\lambda_2-\lambda_1^*}\right)^{-n_1}
\ee^{2\theta_2}\\1
\end{bmatrix},\qquad t\to +\infty,\\
&\begin{bmatrix}\left(\frac{\lambda_2-\lambda_1^*}{\lambda_2-\lambda_1}\right)^{-n_1}
\ee^{2\theta_2}\\1
\end{bmatrix}, \qquad t\to -\infty. \end{split}\right.
\end{split}
\end{equation}
Now we can give the asymptotic expression with two spectral parameters $\lambda_1$ and $\lambda_2$.
\begin{theorem}\label{theorem2}
Suppose there are two spectral parameters $\lambda_1, \lambda_2$ and $v_1<v_2$, then the long-time asymptotic behavior of $q^{[N]}$ is
\small
\begin{equation}\label{eq:qn2}
\begin{split}
&q^{[N]}_{\pm}{=}\sum_{\kappa{=}1}^{\left[\frac{n_1+1}{2}\right]}2\lambda_{1I}
\sech\left(2\lambda_{1I}s_{1, l,\pm}{\pm}\Delta_{\lambda_1, \lambda_2,l}^{[\kappa]}-2a_{1}^{[0]}\right)\ee^{-2\ii{\rm Im}(\theta_1)\pm\ii\frac{n_1+1-2\kappa}{2}
\left(\arg\left(\frac{\mathscr{C}_1}{\mathscr{C}_1^*}\right)\right)\pm\ii\frac{n_2}{2}
\left(\arg\left(\frac{\lambda_1-\lambda_2^*}{\lambda_1-\lambda_2}
\frac{\lambda_1^*-\lambda_2^*}{\lambda_1^*-\lambda_2}\right)\right)+\ii\pi\left(\frac{1}{2}+\kappa\right)}\\&
{+}\sum_{\kappa=1}^{\left[\frac{n_1}{2}\right]}2\lambda_{1I}
 \sech\left(2\lambda_{1I}s_{1, r,\pm}\mp\Delta_{\lambda_1, \lambda_2, r}^{[\kappa]}-2a_{1}^{[0]}\right)\ee^{-2\ii{\rm Im}(\theta_1)\mp\ii\frac{n_1+1-2\kappa}{2}
\left(\arg\left(\frac{\mathscr{C}_1}{\mathscr{C}_1^*}\right)\right)\pm\ii\frac{n_2}{2}
\left(\arg\left(\frac{\lambda_1-\lambda_2^*}{\lambda_1-\lambda_2}
\frac{\lambda_1^*-\lambda_2^*}{\lambda_1^*-\lambda_2}\right)\right)+\ii\pi\left(n_1-\frac{1}{2}+\kappa\right)}\\
 &{+}\sum_{\kappa{=}1}^{\left[\frac{n_2+1}{2}\right]}2\lambda_{2I}
\sech\left(2\lambda_{2I}s_{2, l,\pm}{\pm}\Delta_{\lambda_2, \lambda_1, l}^{[\kappa]}-2a_{2}^{[0]}\right)\ee^{-2\ii{\rm Im}(\theta_2)\pm\ii\frac{n_2+1-2k}{2}
\left(\arg\left(\frac{\mathscr{C}_2}{\mathscr{C}_2^*}\right)\right)\pm\ii\frac{n_1}{2}
\left(\arg\left(\frac{\lambda_2-\lambda_1}{\lambda_2-\lambda_1^*}
\frac{\lambda_2^*-\lambda_1}{\lambda_2^*-\lambda_1^*}\right)\right)+\ii\pi\left(\frac{1}{2}+\kappa\right)}\\&
{+}\sum_{\kappa=1}^{\left[\frac{n_2}{2}\right]}2\lambda_{2I}
 \sech\left(2\lambda_{2I}s_{2, r,\pm}\mp\Delta_{\lambda_2, \lambda_1, r}^{[\kappa]}-2a_{2}^{[0]}\right)
 \ee^{-2\ii{\rm Im}(\theta_2)\mp\ii\frac{n_2+1-2\kappa}{2}
\left(\arg\left(\frac{\mathscr{C}_2}{\mathscr{C}_2^*}\right)\right)\pm\ii\frac{n_1}{2}
\left(\arg\left(\frac{\lambda_2-\lambda_1}{\lambda_2-\lambda_1^*}
\frac{\lambda_2^*-\lambda_1}{\lambda_2^*-\lambda_1^*}\right)\right)+\ii\pi\left(n_2-\frac{1}{2}+\kappa\right)}\\&+O(\log(t)/t).
\end{split}
\end{equation}
\normalsize
where
\begin{equation*}
\begin{split}
&\Delta_{\lambda_i, \lambda_j, l}^{[\kappa]}=\log\left(\frac{(n_i-\kappa)!}{(k-1)!}\right)-(n_i+1-2\kappa)
\left(\log\left(2\lambda_{iI}|\mathscr{C}_i|\right)\right)
- n_j\left(\text{\rm{sgn}}(j-i)\right)
\log\left|\frac{\lambda_i-\lambda_j}{\lambda_i-\lambda_j^*}\right|\\
&\Delta_{\lambda_i, \lambda_j, r}^{[\kappa]}=\log\left(\frac{(n_i-\kappa)!}{(\kappa-1)!}\right)-(n_i+1-2\kappa)
\left(\log\left(2\lambda_{iI}|\mathscr{C}_i|\right)\right)
+ n_j\left(\text{\rm {sgn}}(j-i)\right)
\log\left|\frac{\lambda_i-\lambda_j}{\lambda_i-\lambda_j^*}\right|
\end{split}
\end{equation*}
and $\rm{sgn}$(x) is a sign function defined by
$${\rm sgn}(x):=\left\{\begin{split}
&-1,\quad &if\quad  x<0,\\
&0, \quad  &if\quad  x=0,\\
&1, \quad  &if \quad x>0.
\end{split}\right.$$
\end{theorem}
\begin{proof}
Based on the analysis in Lemma \ref{lemma-denom} and the Lemma \ref{lemma2spectral}, we give the asymptotics behavior when $t\to+\infty$. The other case can be given similarly. If the high-order soliton moves along the trajectory of soliton$_1$, then the new vector $\Phi_{1}^{[0]}$ is defined as Eq.\eqref{eq:phi1-asymptotics}. Under this condition, the leading order terms of the denominator becomes
\begin{equation}
  \begin{split}
  \sum_{j{=}0}^{n_1{-}1}\ee^{2(n_1{-}j)(\theta_1{+}\theta_1^*)}\left(
  \left|\frac{\lambda_1{-}\lambda_2}{\lambda_1{-}\lambda_2^*}\right|^{2n_2(n_1{-}j)}\mathscr{A}_j(x,t)\right)
  +\mathscr{A}_0{+}{\rm l.o.t}.
  \end{split}
  \end{equation}
Similarly, the leading order of the numerator is
\begin{equation}
\begin{split}
&\sum_{j=0}^{n_1-1}(-1)^{j+1}\ee^{2\theta_1^*}\ee^{2(n_1-j-1)(\theta_1+\theta_1^*)}
\left(\frac{\lambda_1^*-\lambda_2^*}{\lambda_1^*-\lambda_2}\right)^{n_{2}(n_1-j)}
\left(\frac{\lambda_1-\lambda_2}{\lambda_1-\lambda_2^*}\right)^{n_{2}(n_1-j-1)}\mathscr{B}_{j}(x,t)+{\rm l.o.t.}
\end{split}
\end{equation}
If the soliton moves along the following characteristic curve
\begin{equation*}
\begin{split}
x&=s_{1, l, +}+\left(\frac{n_1+1-2\kappa}{2\lambda_{1I}}\right)\log(t)+
v_1t,\qquad\left(\kappa=1,2,\cdots, \left[\frac{n_1+1}{2}\right]\right),\\
x&=s_{1, r, +}-\left(\frac{n_1+1-2\kappa}{2\lambda_{1I}}\right)\log(t)+
v_1t,\qquad\left(\kappa=1,2,\cdots, \left[\frac{n_1}{2}\right]\right),
\end{split}
\end{equation*}
then the long-time asymptotics of high-order soliton $q^{[N]}$ is
\small
\begin{equation}
\begin{split}
q^{[N]}&{=}2\lambda_{1I}
\sech\left(2\lambda_{1I}s_{1, l,+}{+}\Delta_{\lambda_1, \lambda_2,l}^{[\kappa]}{-}2a_{1}^{[0]}\right)\ee^{-2\ii{\rm Im}(\theta_1){+}\ii\frac{n_1{+}1{-}2\kappa}{2}
\left(\arg\left(\frac{\mathscr{C}_1}{\mathscr{C}_1^*}\right)\right){+}\ii\frac{n_2}{2}
\left(\arg\left(\frac{\lambda_1{-}\lambda_2^*}{\lambda_1{-}\lambda_2}
\frac{\lambda_1^*{-}\lambda_2^*}{\lambda_1^*{-}\lambda_2}\right)\right){+}\ii\pi\left(\frac{1}{2}+\kappa \right)}{+}O(\log(t)/t),\\&
\text{or}\\
q^{[N]}&{=}2\lambda_{1I}
 \sech\left(2\lambda_{1I}s_{1, r,+}{-}\Delta_{\lambda_1, \lambda_2, r}^{[\kappa]}{-}2a_{1}^{[0]}\right)\ee^{-2\ii{\rm Im}(\theta_1){-}\ii\frac{n_1{+}1{-}2\kappa}{2}
\left(\arg\left(\frac{\mathscr{C}_1}{\mathscr{C}_1^*}\right)\right){+}\ii\frac{n_2}{2}
\left(\arg\left(\frac{\lambda_1{-}\lambda_2^*}{\lambda_1{-}\lambda_2}
\frac{\lambda_1^*{-}\lambda_2^*}{\lambda_1^*{-}\lambda_2}\right)\right){+}\ii\pi\left(n_1-\frac{1}{2}+\kappa \right)}{+}O(\log(t)/t).
\end{split}
\end{equation}
\normalsize
Similarly, if the soliton moves along another characteristic curve
\begin{equation*}
\begin{split}
x&=s_{2, l, +}+\left(\frac{n_2+1-2\kappa}{2\lambda_{2I}}\right)\log(t)+
v_2t,\left(\kappa=1,2,\cdots, \left[\frac{n_2+1}{2}\right]\right),
\\x&=s_{2, r, +}-\left(\frac{n_2+1-2\kappa}{2\lambda_{2I}}\right)\log(t)+
v_2t,\left(\kappa=1,2,\cdots, \left[\frac{n_2}{2}\right]\right),
\end{split}
\end{equation*}
the asymptotic expression is
\small
\begin{equation}
\begin{split}
q^{[N]}&{=}2\lambda_{2I}
\sech\left(2\lambda_{2I}s_{2, l,+}{+}\Delta_{\lambda_2, \lambda_1, l}^{[\kappa]}{-}2a_{2}^{[0]}\right)\ee^{-2\ii{\rm Im}(\theta_2)+\frac{n_2+1-2\kappa}{2}
\left(\arg\left(\frac{\mathscr{C}_2}{\mathscr{C}_2^*}\right)\right)+\ii\frac{n_1}{2}
\left(\arg\left(\frac{\lambda_2{-}\lambda_1}{\lambda_2{-}\lambda_1^*}
\frac{\lambda_2^*{-}\lambda_1}{\lambda_2^*{-}\lambda_1^*}\right)\right){+}\ii\pi\left(\frac{1}{2}+\kappa \right)}{+}O(\log(t)/t),\\&
\text{or}\\
q^{[N]}&{=}2\lambda_{2I}
 \sech\left(2\lambda_{2I}s_{2, r,+}{-}\Delta_{\lambda_2, \lambda_1, r}^{[\kappa]}{-}2a_{2}^{[0]}\right)
 \ee^{-2\ii{\rm Im}(\theta_2){-}\frac{n_2{+}1{-}2\kappa}{2}
\left(\arg\left(\frac{\mathscr{C}_2}{\mathscr{C}_2^*}\right)\right){+}\ii\frac{n_1}{2}
\left(\arg\left(\frac{\lambda_2{-}\lambda_1}{\lambda_2{-}\lambda_1^*}
\frac{\lambda_2^*{-}\lambda_1}{\lambda_2^*{-}\lambda_1^*}\right)\right){+}\ii\pi\left(n_2{-}\frac{1}{2}{+}\kappa \right)}{+}O(\log(t)/t).
\end{split}
\end{equation}
\normalsize
If the high-order soliton moves along the other curves $x=s+vt+\beta\log(t)$ with $v\neq v_1, v\neq v_2 $ or $\beta\neq\frac{n_1+1-2\kappa}{2\lambda_{1I}}, \beta\neq\frac{n_2+1-2\kappa}{2\lambda_{2I}} $, then its asymptotic solution is
\begin{equation}
q^{[N]}=O(\ee^{-2a|v-b|t}),
\end{equation}
and
\begin{equation}
q^{[N]}=O(t^{-1}),
\end{equation}
respectively, where $a=\min(\lambda_{1I}, \lambda_{2I}), b=\min(|v-v_1|, |v-v_2|)$, both of which are vanishing when $t\to\infty$. The asymptotic behavior of $q^{[N]}$ when $t\to-\infty$ can also be given in a similar method. Then the global asymptotics of $q^{[N]}$ can be given in Eq.\eqref{eq:qn2}. It completes the proof.
\end{proof}

In this case, we only present one numerical figure to verify the above asymptotic expressions in Fig.\ref{fig:2-spectral-1}.
\begin{figure}[!h]
{
\includegraphics[height=0.25\textwidth]{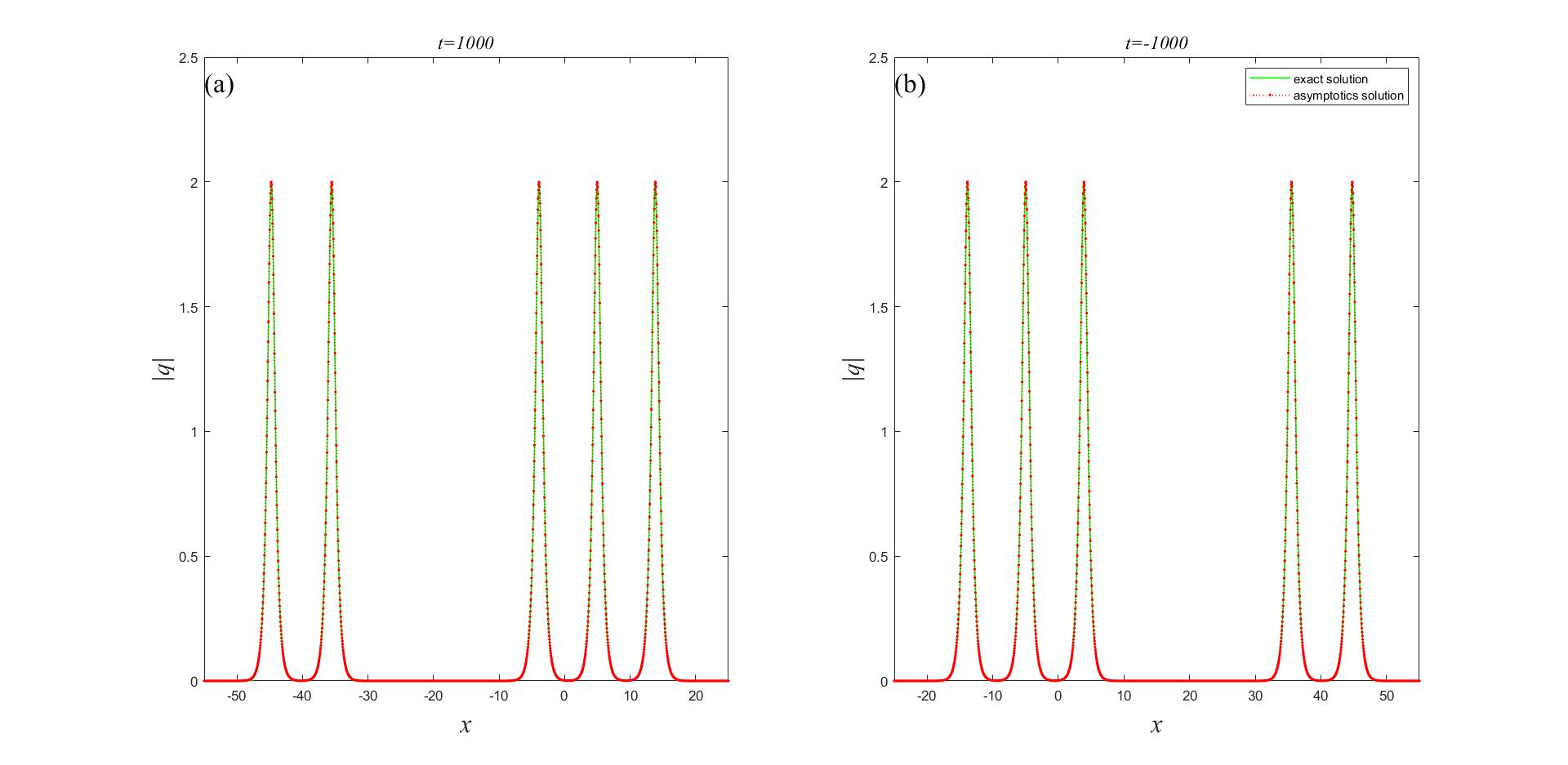}}
\caption{\small The comparison between the exact solution and asymptotic solution by choosing the parameters $\lambda_1=1/9+\ii, \lambda_2=1/8+\ii, \delta=1/8, \gamma=1/2, a_{1}^{[0]}=a_{1}^{[1]}=a_{1}^{[2]}=a_{2}^{[0]}=a_{2}^{[1]}=0, \xi=x-7t/27$, their orders are $n_1=3$, $n_2=2.$}\label{fig:2-spectral-1}
\end{figure}

In Fig.\ref{fig:2-spectral-1}, we choose the parameters $\lambda_1=1/9+\ii, \lambda_2=1/8+\ii, \delta=1/8, \gamma=1/2, a_{1}^{[0]}=a_{1}^{[1]}=a_{1}^{[2]}=a_{2}^{[0]}=a_{2}^{[1]}=0$, then
\begin{equation*}
\begin{split}
\theta_1&=-x+\frac{7}{27}t+\ii\left(\frac{x}{9}-\frac{841t}{729}-\frac{\pi}{4}\right),\\
\theta_2&=-x+\frac{7}{27}t-\frac{113}{3456}t+\ii\left(\frac{x}{8}-\frac{1199t}{1024}-\frac{\pi}{4}\right).
\end{split}
\end{equation*}
It is clear $v_1=\frac{7}{27}, v_2=\frac{29}{128}, v_1>v_2$. If the soliton moves along the trajectory of soliton$_1$, we have \begin{equation}
\begin{split}
&\theta_2\to-\infty, \quad t\to+\infty,\\
&\theta_2\to+\infty, \quad t\to-\infty.\\
\end{split}
\end{equation}
With the asymptotics in Eq.\eqref{eq:phi2-asymptotics}, we know $\Phi_{2}^{[0]}$ has asymptotically form of $\begin{bmatrix}0\\1
\end{bmatrix}$ and $\begin{bmatrix}1\\0
\end{bmatrix}$ as $t\to +\infty$ and $-\infty$ respectively.
When $t>0$, soliton$_2$ is located on the left of soliton$_1$, but when $t<0$, soliton$_1$ is in the left because soliton$_1$ moves faster. Conversely, if the soliton moves along the trajectory of soliton$_2$, we know $\Phi_{1}^{[0]}$ has an opposite asymptotics, which approaches $\begin{bmatrix}1\\0
\end{bmatrix}$ and $\begin{bmatrix}0\\1
\end{bmatrix}$ when $t\to+\infty$ and $-\infty$ respectively.
It is interesting that the soliton either decays or grows exponentially as $t\to\pm\infty$, along any direction except its own trajectory. In order to observe these five solitons when $t$ is large, we choose the velocity $v_1, v_2$ be close. Otherwise, one soliton will go far away to the other ones when $t$ is large, it becomes difficult to see all of the solitons in a short range in space.

Next, we will give a general asymptotic expression with $k$ spectral parameters $\lambda_1, \lambda_2, \cdots, \lambda_k$.

\begin{theorem}\label{theorem-k}
If there are $k$ spectral parameters $\lambda_1, \lambda_2, \cdots, \lambda_k$ with the order $n_1, n_2, \cdots, n_k$ respectively. Suppose their velocities satisfying $v_1<v_2<\ldots<v_k$, then the long-time asymptotics of $q^{[N]}$ is
\scriptsize
\begin{equation}\label{eq:qnk}
\begin{split}
&q_{\pm}^{[N]}{=}\sum_{i{=}1}^{k}\sum_{\kappa{=}1}^{\left[\frac{n_i{+}1}{2}\right]}2\lambda_{iI}
\sech\left(2\lambda_{iI}s_{i,l,\pm}\pm\Delta_{\lambda_i,\cdots,\lambda_k,\lambda_1,\cdots,\lambda_{i-1},  l}^{[\kappa]}-2a_{i}^{[0]}\right)\ee^{-2\ii{\rm Im}(\theta_i)\pm\ii\frac{n_i{+}1{-}2\kappa}{2}\left(\arg\left(\frac{\mathscr{C}_i}{\mathscr{C}_i^*}\right)\right)
\pm\ii\sum\limits_{j{=}1, j\neq i}^{k}\frac{n_j}{2}{\rm sgn}(j-i)\left(\arg\left(\frac{\lambda_i{-}\lambda_j^*}{\lambda_i{-}\lambda_j}
\frac{\lambda_i^*{-}\lambda_j^*}{\lambda_i^*{-}\lambda_j}\right)\right){+}\ii\pi\left(\frac{1}{2}{+}\kappa\right)}
\\&{+}\sum_{i{=}1}^{k}\sum_{\kappa{=}1}^{\left[\frac{n_i}{2}\right]}
2\lambda_{iI}
 \sech\left(2\lambda_{iI}s_{i,r,\pm}{\mp}\Delta_{\lambda_i,\cdots,\lambda_{k},\lambda_1,\cdots,\lambda_{i-1}, r}^{[\kappa]}-2a_{i}^{[0]}\right)
 \ee^{-2\ii{\rm Im}(\theta_i)\mp\ii\frac{n_i{+}1{-}2\kappa}{2}\left(\arg\left(\frac{\mathscr{C}_i}{\mathscr{C}_i^*}\right)\right)
\pm\ii\sum\limits_{j{=}1, j\neq i}^{k}\frac{n_j}{2}{\rm sgn}(j-i)\left(\arg\left(\frac{\lambda_i{-}\lambda_j^*}{\lambda_i{-}\lambda_j}
\frac{\lambda_i^*{-}\lambda_j^*}{\lambda_i^*{-}\lambda_j}\right)\right){+}\ii\pi\left(n_i{-}\frac{1}{2}{+}\kappa \right)}\\
 &+O(\log(t)/t)
 \end{split}
\end{equation}
\normalsize
where
\small
\begin{equation*}
\begin{split}
&\Delta_{\lambda_i,\cdots,\lambda_{k},\lambda_1,\cdots,\lambda_{i-1},l}^{[\kappa]}
{=}\log\left(\frac{(n_i-\kappa)!}{(\kappa-1)!}\right){-}(n_i+1-2\kappa)
\left(\log\left(2\lambda_{iI}\right)+\log\left(|\mathscr{C}_{i}|\right)\right){-}\sum_{j=1, j\neq i}^{k}n_j{\rm sgn}(j{-}i)
\log\left|\frac{\lambda_i{-}\lambda_j}{\lambda_i{-}\lambda_j^*}\right|
\\
&\Delta_{\lambda_i,\cdots,\lambda_k,\lambda_1,\cdots, \lambda_{i-1}, r}^{[\kappa]}
{=}\log\left(\frac{(n_i-\kappa)!}{(k-1)!}\right){-}(n_i+1-2\kappa)
\left(\log\left(2\lambda_{iI}\right){+}\log\left(|\mathscr{C}_{i}|\right)\right){+}\sum_{j=1, j\neq i}^{k}n_j{\rm sgn}(j-i)
\log\left|\frac{\lambda_i{-}\lambda_j}{\lambda_i{-}\lambda_j^*}\right|\\
&\theta_i:=\theta^{[i]}\big|_{\lambda=\lambda_i}=\ii\lambda_i(x+(2\gamma\lambda_i+4\delta\lambda_i^2)t)
+a_i^{[0]}-\frac{\ii}{4}\pi,\quad \lambda_i=\lambda_{iR}+\ii\lambda_{iI}.
\end{split}
\end{equation*}
\end{theorem}
\begin{proof}
\normalsize
Similar to the analysis with two spectral parameters, we can give its asymptotics along the trajectory of every soliton$_i$, $(i=1,2,\cdots, k)$. If the soliton moves along the trajectory soliton$_i$, that is ${\rm Re}(\theta_i)=O(1)$. When $t\to+\infty$, we have
\begin{equation*}\theta_j\to
\left\{
\begin{split}
&-\infty,\quad j<i,\\
&+\infty,\quad j>i.
\end{split}
\right.
\end{equation*}
Then we can define a new $\Phi_{i}^{[0]}$ as
$$\Phi_{i}^{[0]}=\begin{bmatrix}\prod\limits_{j=1, j\neq i}\limits^{k}\left(\frac{\lambda_i-\lambda_j}{\lambda_i-\lambda_j^*}\right)^{n_j{\rm sgn}(j-i)}\ee^{2\theta_i}\\
1
\end{bmatrix}$$
Based on the analysis in Theorem \ref{theorem2}, if the high-order soliton moves along the characteristic curve
\begin{equation}
\begin{split}
x&=s+v_it+\frac{n_i+1-2\kappa}{2\lambda_{iI}}\log(|t|),\qquad \left(i=1,2,\cdots, k , \kappa=1,2,\cdots, \left[\frac{n_i+1}{2}\right]\right),\\
x&=s+v_it-\frac{n_i+1-2\kappa}{2\lambda_{iI}}\log(|t|),\qquad \left(i=1,2,\cdots, k, \kappa=1,2,\cdots, \left[\frac{n_i}{2}\right]\right),
\end{split}
\end{equation}
then the asymptotics of $q^{[N]}$ is
\begin{equation}
\begin{split}
q^{[N]}&=2\lambda_{iI}
\sech\left(2\lambda_{iI}s_{i,l,\pm}\pm\Delta_{\lambda_i,\cdots,\lambda_k,\lambda_1,\cdots,\lambda_{i-1},  l}^{[\kappa]}{-}2a_{i}^{[0]}\right)\\
\cdot&\ee^{-2\ii{\rm Im}(\theta_i)\pm\ii\frac{n_i{+}1{-}2\kappa}{2}\left(\arg\left(\frac{\mathscr{C}_i}{\mathscr{C}_i^*}\right)\right)
\pm\ii\sum\limits_{j{=}1, j\neq i}^{k}\frac{n_j}{2}{\rm sgn}(j{-}i)\left(\arg\left(\frac{\lambda_i{-}\lambda_j^*}{\lambda_i{-}\lambda_j}
\frac{\lambda_i^*{-}\lambda_j^*}{\lambda_i^*{-}\lambda_j}\right)\right){+}\ii\pi\left(\frac{1}{2}+\kappa \right)}{+}O(\log(t)/t)\\
&\text{or}\\
q^{[N]}&=2\lambda_{iI}
 \sech\left(2\lambda_{iI}s_{i,r,\pm}{\mp}\Delta_{\lambda_i,\cdots,\lambda_{k},\lambda_1,\cdots,\lambda_{i-1}, r}^{[\kappa]}{-}2a_{i}^{[0]}\right)
\\\cdot& \ee^{-2\ii{\rm Im}(\theta_i)\mp\ii\frac{n_i{+}1{-}2\kappa}{2}\left(\arg\left(\frac{\mathscr{C}_i}{\mathscr{C}_i^*}\right)\right)
\pm\ii\sum\limits_{j{=}1, j\neq i}^{k}\frac{n_j}{2}{\rm sgn}(j{-}i)\left(\arg\left(\frac{\lambda_i{-}\lambda_j^*}{\lambda_i{-}\lambda_j}
\frac{\lambda_i^*{-}\lambda_j^*}{\lambda_i^*{-}\lambda_j}\right)\right){+}\ii\pi\left(n_i{-}\frac{1}{2}+\kappa \right)}{+}O(\log(t)/t)
\end{split}
\end{equation}
Otherwise, if the high-order soliton moves along the other curves $x=s+vt+\beta\log(t)$ with $v\neq v_i,$ or $\beta\neq\frac{n_i+1-2\kappa}{2\lambda_{iI}}(i=1,2,\cdots, k)$, then its asymptotics is
\begin{equation}
q^{[N]}=O(\ee^{-2a|v-b|t}),
\end{equation}
and
\begin{equation}
q^{[N]}=O(t^{-1}),
\end{equation}
respectively, where $a=\min(\lambda_{iI},(i=1,2,\cdots, k)), b=\min(|v-v_i|)(i=1,2,\cdots, k)$, which is an exponential decay or algebraic decay. The asymptotic behavior of $q^{[N]}$ when $t\to-\infty$ can also be calculated in a similar method. Then the long-time asymptotics of $q^{[N]}$ can be given in Eq.\eqref{eq:qnk}. It completes the proof.
\end{proof}

Now we choose three spectral parameters $\lambda_1, \lambda_2, \lambda_3$ to verify this asymptotic behavior Eq.\eqref{eq:qnk} in Fig.\ref{fig:3-spectral}.

\begin{figure}[!h]
{
\includegraphics[height=0.35\textwidth]{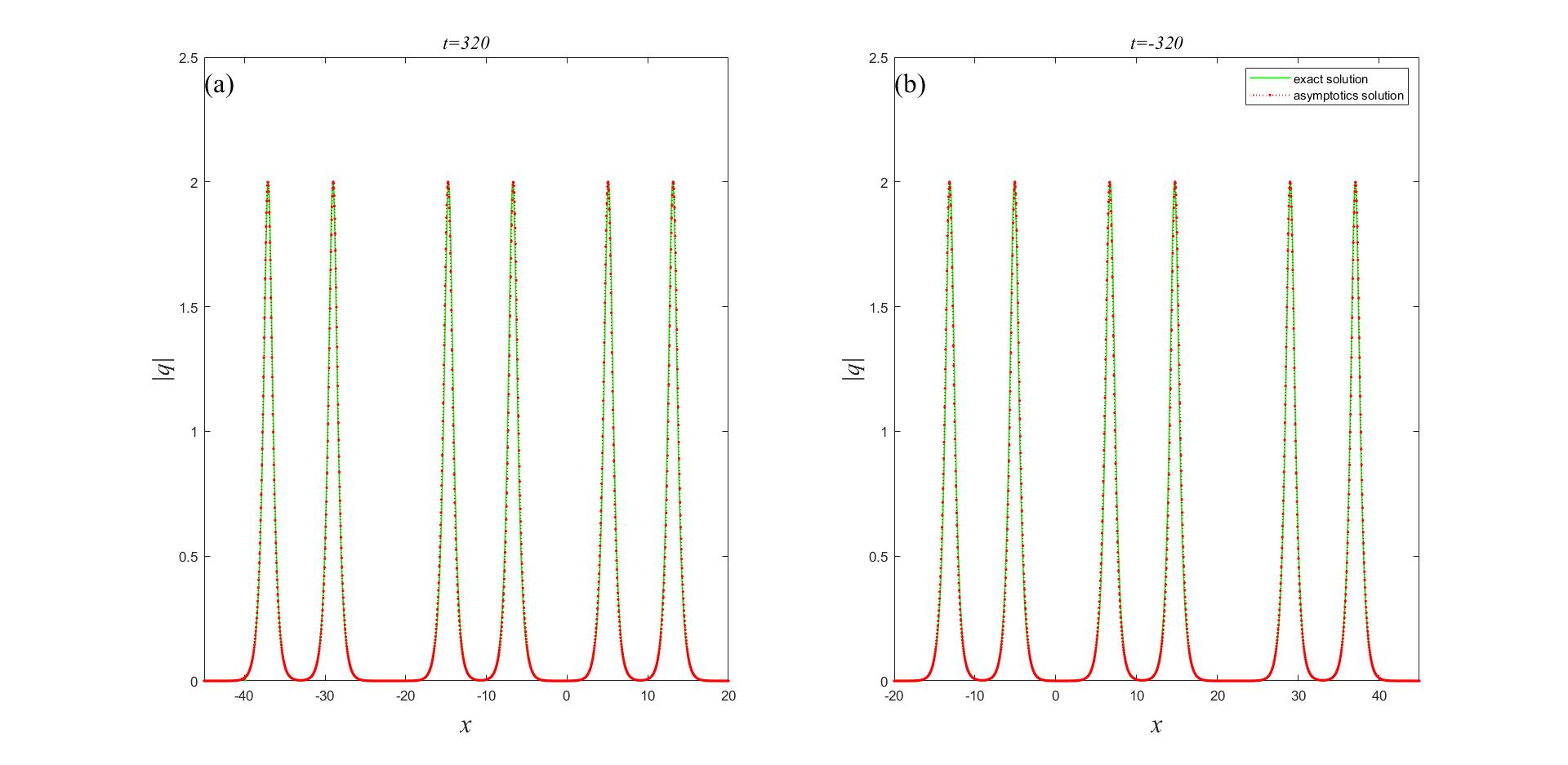}}
\caption{\small The comparison between the exact solution and asymptotic solution by choosing the parameters $\lambda_1=1/9+\ii, \lambda_2=1/8+\ii, \lambda_3=1/7+\ii, \delta=1/8, \gamma=1/2, a_{1}^{[0]}=a_{1}^{[1]}=a_{2}^{[0]}=a_{2}^{[1]}=a_{3}^{[0]}=a_{3}^{[1]}=0, \xi=x-7t/27, n_1=n_2=n_3=2$.}\label{fig:3-spectral}
\end{figure}
In this case, we still set the velocities $v_1, v_2, v_3$ be closed so as to observe them in a short range in space. Their velocities satisfy $v_1>v_2>v_3$, so soliton$_2$ is always located in the middle of soliton$_1$ and soliton$_3$, whose dynamical behavior is consistent with the theorem \ref{theorem-k}.
\section{Conclusion}

\lm{In this work, we analyze the long-time asymptotics for the high-order soliton of Hirota equation based on the Darboux transformation. Actually, we propose a method to strictly analyze the high-order soliton for the integrable equations under the framework of Darboux transformation essentially.}

\lm{Firstly, we reexamine the AKNS hierarchy with $su(2)$ symmetry, which can be used to derive the Hirota equation. Based on the Darboux matrix for the high-order soliton, we construct three RHPs. The first two RHPs are related with the robust inverse scattering transform. The third one establishes the relation with the classical inverse scattering transform. The second RHP can be used to tackle with the far-field asymptotic analysis for large high-order soliton (as the order $n\to\infty$) under the approach proposed in the literature by one of the author with Bilman and Miller \cite{Peter-Duke-2019}. As we can see that the Riemann-Hilbert representation for the Darboux matrix of high-order soliton is much simpler than the Darboux matrix itself. Thus it is natural to consider the asymptotics by the RHP as $t\to\infty$. Unluckily, we fail as we try to solve it by the Deift-Zhou method directly. However we solve it throughout the determinant formulas successfully.}

Secondly, the high-order soliton with single spectral parameter $\lambda_1$ is analyzed by the determinant. The key idea about the asymptotic analysis is looking for the characteristic curves. If the soliton moves along this characteristic curve, then the high-order soliton can reduce to a single soliton. Otherwise, it has the vanishing limitation. We obtain the exact high-order soliton and give its leading order term from the solitons determinant directly. With the aid of the leading order term, we give the asymptotic behavior with single spectral parameter $\lambda_1$. In this case, we present a detailed \lm{analysis or strict proof} about how to obtain the asymptotic expression for $t\to\infty$.

Furthermore, we give the asymptotic expression to high-order soliton with two spectral parameters $\lambda_1$ and $\lambda_2$. In this case, by using the asymptotics of Darboux matrix, we convert the problem with two spectral parameters to the problem with single spectral parameter. Thus the asymptotics to high-order soliton with two spectral parameters can be given with the result of theorem \ref{theorem1}. Based on this analysis, we also give the general asymptotic expression to high-order soliton with $k$ spectral parameters. For each case, we give the numeric plotting to verify the \lm{result} of the theorems.

In this paper, we just consider the finite order soliton and analyze the long-time asymptotics. \lm{Under one spectral parameter $\lambda_1$, the high-order soliton represents the interaction between $N$-solitons with the same amplitude but having a particular chirp \cite{Gagnon-OL-1994}. In contrast to the multi-pole solutions with distinct spectral parameters, the high-order soliton with single pole as well as its asymptotic behavior has received less attention. Thus this topic in this paper is worthy of study in the integrable field. More significantly, the analysis method proposed in this work is only related to the Darboux matrix, which can be regarded as efficient tool to analyze the asymptotic behavior directly.} Thus this method can also be readily extended to the other integrable models, such as the sine-Gordan equation, modified KdV equation, derivative NLS equation and so on. Very recently, there exist some progresses on the large order soliton for the NLS equation \cite{Peter-Duke-2019,Deniz-JNS-2019,Deniz-arXiv-2019}. For the large order soliton, the complete description on the high-order soliton in different regions \lm{(as the order $n\to\infty$)} is also an interesting problem, which is deserved to explore by the RHP in the \lm{second} section in the future studying.

\section*{Appendix A}
\label{App-higher-soliton}
In this appendix, we give two proof to lemma \ref{lemma-denom} and lemma \ref{lemmamatch}. We first present the proof to the asymptotic behavior of the denominator and the numerator in lemma \ref{lemma-denom}:
\begin{proof} Through the formula \eqref{eq:b1-representation} and the determinant formulas $\det(\mathbf{K}\mathbf{L})=\det\left(\begin{matrix}
	0&\mathbf{K} \\
-\mathbf{L}&\mathbb{I}\\
	\end{matrix}\right)$,
the determinant of the denominator $\mathbf{F}_1^{\dag}\mathbf{B}_1\mathbf{F}_1+\mathbf{B}_1$ can be rewritten as
\begin{equation}
\begin{split}
&|\mathbf{F}_1^{\dag}\mathbf{B}_1\mathbf{F}_1+\mathbf{B}_1|=
\left|\left(-\frac{\ii}{2\lambda_{1I}}\right)\left(\mathbf{F}_1^{\dag}\mathbf{D}_1\mathbf{S}_N^{\dag},\mathbf{D}_1\mathbf{S}_N^{\dag}\right)\begin{bmatrix}
\mathbf{S}_N\mathbf{D}_1\mathbf{F}_1\\
\mathbf{S}_N\mathbf{D}_1\\
\end{bmatrix}\right|\\
&=\left(-\frac{\ii}{2\lambda_{1I}}\right)^N\left|\begin{matrix}0&\mathbf{F}_1^{\dag}\mathbf{D}_1\mathbf{S}_N^{\dag}&\mathbf{D}_1\mathbf{S}_N^{\dag}\\
-\mathbf{S}_N\mathbf{D}_1\mathbf{F}_1&\mathbb{I}_N& 0\\
-\mathbf{S}_N\mathbf{D}_1&0&\mathbb{I}_N
\end{matrix}\right|\\
&=\left(-\frac{\ii}{2\lambda_{1I}}\right)^N\left|\begin{matrix}0&\mathbf{F}_1^{\dag}\mathbf{D}_1+\text{l.o.t}&\mathbf{D}_1\mathbf{S}_N^{\dag}\\
-\mathbf{D}_1\mathbf{F}_1+\text{l.o.t}&\mathbb{I}_N&0&\\
-\mathbf{S}_N\mathbf{D}_1&0&\mathbb{I}_N
\end{matrix}\right|,
\end{split}
\end{equation}
\normalsize
since the matrix $\mathbf{S}_N\mathbf{D}_1\mathbf{F}_1=\mathbf{D}_1\mathbf{F}_1+\text{l.o.t}$ and the leading order term of every entry $f_{1, j}$ in $\mathbf{F}_1$ is $\frac{\vartheta_{1,0}^{j}}{j!}$. We calculate this determinant by the general Laplace expansion method. In this case, we choose the first $N$ column, the determinant has $\binom{2N}{N}$ terms. The leading order term is a special choice, which starts from $N+1$ row to $N+j$ row and ends up with $2N+1$ row to $3N-j$ row, where $(j=0,\cdots, N)$.  It needs to be emphasized that $j=0$ indicates $2N+1$ row to $3N$ row, and $j=N$ indicates $N+1$ row to $2N$ row. As to the complementary cofactor matrix, we continue to calculate it with the Laplace expansion method. For simplicity, we choose the first $N$ row, and the determinant can be given easily. Next, we give a detailed calculation about the leading order coefficient for $\ee^{2(N-j)(\theta_1^*+\theta_1)}$. Following the given rule, the coefficient is
\scriptsize
\begin{equation}\label{eq:lead-order1}
\begin{split}
\scriptsize
&\left|\begin{array}{cccc:ccc}-1&-(\vartheta_{1,0})&\cdots&-\frac{\vartheta_{1,0}^{j-1}}{(j-1)!}
&-\frac{\vartheta_{1,0}^{j}}{j!}&\cdots&-\frac{\vartheta_{1,0}^{N-1}}{(N-1)!}\\
0&-\left(-\frac{1}{2\lambda_{1I}}\right)&\cdots&-\left(-\frac{1}{2\lambda_{1I}}\right)
\frac{\vartheta_{1,0}^{j-2}}{(j-2)!}&-\left(-\frac{1}{2\lambda_{1I}}\right)\frac{\vartheta_{1,0}^{j-1}}{(j-1)!}
&\cdots&-\left(-\frac{1}{2\lambda_{1I}}\right)\frac{\vartheta_{1,0}^{N-2}}{(N-2)!}\\
\vdots&\vdots&\vdots&\vdots&\vdots&\vdots&\vdots\\
0&0&\cdots&-\left(-\frac{1}{2\lambda_{1I}}\right)^{N-j-1}\frac{\vartheta_{1,0}^{2j-N}}{(2j-N)!}
&-\left(-\frac{1}{2\lambda_{1I}}\right)^{N-j-1}\frac{\vartheta_{1,0}^{2j-N+1}}{(2j-N+1)!}&\cdots&-\left(-\frac{1}{2\lambda_{1I}}\right)^{N-j-1}
\frac{\vartheta_{1,0}^{j}}{j!}\\\cdashline{1-7}[4pt/5pt]
-\binom{0}{0}\left(-\frac{1}{2\lambda_{1I}}\right)^0&-\binom{1}{0}\left(-\frac{1}{2\lambda_{1I}}\right)^1
&\cdots&-\binom{j-1}{0}\left(-\frac{1}{2\lambda_{1I}}\right)^{j-1}&-\binom{j}{0}\left(-\frac{1}{2\lambda_{1I}}\right)^{j}
&\cdots&-\binom{N-1}{0}\left(-\frac{1}{2\lambda_{1I}}\right)^{N-1}\\
0&-\binom{1}{1}\left(-\frac{1}{2\lambda_{1I}}\right)^1&\cdots&-\binom{j-1}{1}\left(-\frac{1}{2\lambda_{1I}}\right)^{j-1}
&-\binom{j}{1}\left(-\frac{1}{2\lambda_{1I}}\right)^{j}&\cdots&-\binom{N-1}{1}\left(-\frac{1}{2\lambda_{1I}}\right)^{N-1}\\
\vdots&\vdots&\vdots&\vdots&\vdots&\vdots&\vdots\\
0&0&\cdots&-\binom{j-1}{j-1}\left(-\frac{1}{2\lambda_{1I}}\right)^{j-1}&-\binom{j}{j-1}
\left(-\frac{1}{2\lambda_{1I}}\right)^{j}&\cdots&-\binom{N-1}{j-1}\left(-\frac{1}{2\lambda_{1I}}\right)^{N-1}
\end{array}\right|\\
&\cdot\left|\begin{array}{cccc:cccc}1&0&\cdots&0&\binom{0}{0}\left(-\frac{1}{2\lambda_{1I}}\right)^0&0&\cdots&0\\
\vartheta_{1,0}^*&-\frac{1}{2\lambda_{1I}}&\cdots&0&\binom{1}{0}\left(-\frac{1}{2\lambda_{1I}}\right)^1&\binom{1}{1}\left(-\frac{1}{2\lambda_{1I}}\right)^1&\cdots&0\\
\vdots&\vdots&\vdots&\vdots&\vdots&\vdots&\vdots&\vdots\\
\frac{\vartheta_{1,0}^{*,j-1}}{(j-1)!}
&-\frac{1}{2\lambda_{1I}}\frac{\vartheta_{1,0}^{*,j-2}}{(j-2)!}&\cdots&\left(-\frac{1}{2\lambda_{1I}}\right)^{N-j-1}
\frac{\vartheta_{1,0}^{*,2j-N}}{(2j-N)!}&\binom{j-1}{0}\left(-\frac{1}{2\lambda_{1I}}\right)^{j-1}&\binom{j-1}{1}\left(-\frac{1}{2\lambda_{1I}}\right)^{j-1}
&\cdots&\binom{j-1}{j-1}\left(-\frac{1}{2\lambda_{1I}}\right)^{j-1}\\\cdashline{1-8}[4pt/5pt]
\frac{\vartheta_{1,0}^{*,j}}{j!}
&-\frac{1}{2\lambda_{1I}}\frac{\vartheta_{1,0}^{*,j-1}}{(j-1)!}&\cdots&\left(-\frac{1}{2\lambda_{1I}}\right)^{N-j-1}
\frac{\vartheta_{1,0}^{*,2j-n+1}}{(2j-N+1)!}&\binom{j}{0}\left(-\frac{1}{2\lambda_{1I}}\right)^{j}&\binom{j}{1}\left(-\frac{1}{2\lambda_{1I}}\right)^{j-1}
&\cdots&\binom{j}{j-1}\left(-\frac{1}{2\lambda_{1I}}\right)^{j}\\
\vdots&\vdots&\vdots&\vdots&\vdots&\vdots&\vdots&\vdots\\
\frac{\vartheta_{1,0}^{*,N-1}}{(N-1)!}&-\frac{1}{2\lambda_{1I}}\frac{\vartheta_{1,0}^{*,N-2}}{(N-2)!}
&\cdots&\left(-\frac{1}{2\lambda_{1I}}\right)^{N-j-1}\frac{\vartheta_{1,0}^{*,j}}{j!}
&\binom{N-1}{0}\left(-\frac{1}{2\lambda_{1I}}\right)^{N-1}&\binom{N-1}{1}\left(-\frac{1}{2\lambda_{1I}}\right)^{N-1}
&\cdots&\binom{N-1}{j-1}\left(-\frac{1}{2\lambda_{1I}}\right)^{N-1}\end{array}\right|\\
&\cdot(-1)^{2j^2+3N^2+2N}\left(-\frac{\ii}{2\lambda_{1I}}\right)^N+\text{l.o.t}
\end{split}
\end{equation}
\normalsize
Furthermore, we still use the Laplace expansion method to calculate both two determinants in Eq. \eqref{eq:lead-order1}. The leading order term in the first determinant is choosing $1$ to $j$ column, $N-j+1$ row to $N$ row, as the dotted line in the determinant. The second determinant can be given similarly, so Eq. \eqref{eq:lead-order1} becomes
\scriptsize
\begin{equation}\label{eq:numdet}
\begin{split}
&\left(-\frac{\ii}{2\lambda_{1I}}\right)^N\left|\begin{array}{cccc}
\binom{0}{0}\left(-\frac{1}{2\lambda_{1I}}\right)^0&\binom{1}{0}\left(-\frac{1}{2\lambda_{1I}}\right)^1&\cdots&\binom{j-1}{0}\left(-\frac{1}{2\lambda_{1I}}\right)^{j-1}\\
0&\binom{1}{1}\left(-\frac{1}{2\lambda_{1I}}\right)^1&\cdots&\binom{j-1}{1}\left(-\frac{1}{2\lambda_{1I}}\right)^{j-1}\\
\vdots&\vdots&\vdots&\vdots\\
0&0&\cdots&\binom{j-1}{j-1}\left(-\frac{1}{2\lambda_{1I}}\right)^{j-1}
\end{array}\right|\left|\begin{array}{ccc}\frac{\vartheta_{1,0}^{j}}{j!}&\cdots&\frac{\vartheta_{1,0}^{N-1}}{(N-1)!}\\
-\frac{1}{2\lambda_{1I}}\frac{\vartheta_{1,0}^{j-1}}{(j-1)!}&\cdots&-\frac{1}{2\lambda_{1I}}\frac{\vartheta_{1,0}^{N-2}}{(N-2)!}\\
\vdots&\vdots&\vdots\\
\left(-\frac{1}{2\lambda_{1I}}\right)^{N-j-1}\frac{\vartheta_{1,0}^{2j-N+1}}{(2j-N+1)!}&\cdots&\left(-\frac{1}{2\lambda_{1I}}\right)^{N-j-1}
\frac{\vartheta_{1,0}^{j}}{j!}
\end{array}\right|\\&\cdot\left|\begin{array}{cccc}\binom{0}{0}\left(-\frac{1}{2\lambda_{1I}}\right)^{0}&0&\cdots&0\\
\binom{1}{0}\left(-\frac{1}{2\lambda_{1I}}\right)^{1}&\binom{1}{1}\left(-\frac{1}{2\lambda_{1I}}\right)^{1}&\cdots&0\\
\vdots&\vdots&\vdots&\vdots\\
\binom{j-1}{0}\left(-\frac{1}{2\lambda_{1I}}\right)^{j-1}&\binom{j-1}{1}\left(-\frac{1}{2\lambda_{1I}}\right)^{j-1}
&\cdots&\binom{j-1}{j-1}\left(-\frac{1}{2\lambda_{1I}}\right)^{j-1}
\end{array}\right|\left|
\begin{array}{cccc}
\frac{\vartheta_{1,0}^{*,j}}{j!}
&-\frac{1}{2\lambda_{1I}}\frac{\vartheta_{1,0}^{*,j-1}}{(j-1)!}&\cdots&\left(-\frac{1}{2\lambda_{1I}}\right)^{N-j-1}
\frac{\vartheta_{1,0}^{*,2j-N+1}}{(2j-N+1)!}\\
\vdots&\vdots&\vdots&\vdots\\
\frac{\vartheta_{1,0}^{*,N-1}}{(N-1)!}&-\frac{1}{2\lambda_{1I}}\frac{\vartheta_{1,0}^{*,N-2}}{(N-2)!}
&\cdots&\left(-\frac{1}{2\lambda_{1I}}\right)^{N-j-1}\frac{\vartheta_{1,0}^{*,j}}{j!}
\end{array}\right|+{\rm l.o.t}\\=&\left(\left(-\frac{1}{2\lambda_{1I}}\right)^{j(j-1)+(N-j)(N-j-1)}\left(-\frac{\ii}{2\lambda_{1I}}\right)^N|\vartheta_{1,0}|^{2j(N-j)}\right)
\left|\begin{matrix}\frac{1}{j!}&
\frac{1}{(j+1)!}&\cdots&\frac{1}{(N-1)!}\\
\frac{1}{(j-1)!}&\frac{1}{j!}&\cdots&\frac{1}{(N-2)!}\\
\vdots&\vdots&\vdots&\vdots\\
\frac{1}{(2j-N+1)!}&\cdots&\cdots&\frac{1}{j!}
\end{matrix}\right|^2+{\rm l.o.t}.
\end{split}
\end{equation}
\normalsize
With Lemma \ref{lemmaGamma}, we get Eq.\eqref{eq:det-for-denom}.

Next we give the proof to the asymptotics of the numerator, which can be written as
\begin{equation}
\begin{split}
\det{G}&=-(\lambda_1-\lambda_1^*)\ee^{-2\theta_1}\left|
\begin{matrix}
\frac{\ee^{2(\theta_1+\theta_1^*)}}{\lambda_1-\lambda_1^*}&D_{12}&D_{13}&\cdots&D_{1N}\\
\frac{\ee^{2(\theta_1+\theta_1^*)}}{\lambda_1-\lambda_1^*}f_{1,1}^*&D_{22}&D_{23}&\cdots&D_{2N}\\
\vdots&\vdots&\vdots&\vdots&\vdots\\
\frac{\ee^{2(\theta_1+\theta_1^*)}}{\lambda_1-\lambda_1^*}f_{1,N-1}^*&D_{N2}&D_{N3}&\cdots&D_{NN}
\end{matrix}\right|\\
&=-(\lambda_1-\lambda_1^*)\ee^{-2\theta_1}\left|\frac{1}{\lambda_1-\lambda_1^*}
\left(F_1^{\dag}\mathbf{D}_1\mathbf{S}_1^{\dag},\mathbf{D}_1\mathbf{S}_1^{\dag}\right)\begin{bmatrix}
(\mathbf{S}_1\mathbf{D}_1F_1)_{[1]}+\ee^{2\theta_1}(S_1^{-\dag})_1\\
(\mathbf{S}_1\mathbf{D}_1)_{[1]}\\
\end{bmatrix}\right|\\
&=-(\lambda_1-\lambda_1^*)\ee^{-2\theta_1}\left|\frac{1}{\lambda_1-\lambda_1^*}\left(F_1^{\dag}\mathbf{D}_1+\text{l.o.t},\mathbf{D}_1\mathbf{S}_1^{\dag}\right)\begin{bmatrix}
(\mathbf{D}_1F_1)_{[1]}+\ee^{2\theta_1}(S_1^{-\dag})_1+\text{l.o.t}\\
(\mathbf{S}_1\mathbf{D}_1)_{[1]}\\
\end{bmatrix}\right|\\
&=-\left(-\frac{\ii}{2\lambda_{1I}}\right)^{N-1}\ee^{-2\theta_1}\left|\begin{matrix}0&\mathbf{F}_1^{\dag}\mathbf{D}_1+\text{l.o.t}&\mathbf{D}_1\mathbf{S}_N^{\dag}\\
(\mathbf{D}_1F_1)_{[1]}+\ee^{2\theta_1}(S_1^{-\dag})_1+\text{l.o.t}&\mathbb{I}_N&0&\\
-(\mathbf{S}_1\mathbf{D}_1)_{[1]}&0&\mathbb{I}_N
\end{matrix}\right|
\end{split}
\end{equation}
\normalsize
where $(\cdot)_{[1]}$ represents the matrix $\cdot$ by replacing elements of the first column with zeros and the symbol $(\cdot)_1$ denotes taking the first column of the matrix $(\cdot)$ and setting the other elements as zeros. Similar to the analysis in the denominator, we continue to calculate this determinant with the Laplace expansion method. According to the method in the denominator, we give a detailed calculation about the leading order coefficient of $\ee^{2\theta_1^*}\ee^{2(N-j-1)(\theta_1+\theta_1^*)}$, which is
\scriptsize
\begin{equation}\label{eq:dendet}
\begin{split}
&\left|\begin{array}{cccc:ccc}-1&-\vartheta_{1,0}&\cdots&-\frac{\vartheta_{1,0}^{j}}{j!}
&-\frac{\vartheta_{1,0}^{j+1}}{(j+1)!}&\cdots&-\frac{\vartheta_{1,0}^{N-1}}{(N-1)!}\\
1&\frac{1}{2\lambda_{1I}}&\cdots&\frac{1}{2\lambda_{1I}}\frac{\vartheta_{1,0}^{j-1}}{(j-1)!}&\frac{1}{2\lambda_{1I}}\frac{\vartheta_{1,0}^{j}}{(j)!}
&\cdots&\frac{1}{2\lambda_{1I}}\frac{\vartheta_{1,0}^{N-2}}{(N-2)!}\\
\vdots&\vdots&\vdots&\vdots&\vdots&\vdots&\vdots\\\cdashline{1-7}[4pt/5pt]
(-1)^{N-j}&0&\cdots&-\left(-\frac{1}{2\lambda_{1I}}\right)^{N-j-1}\frac{\vartheta_{1,0}^{2j+1-N}}{(2j+1-N)!}
&-\left(-\frac{1}{2\lambda_{1I}}\right)^{N-j-1}\frac{\vartheta_{1,0}^{2j-N+2}}{(2j-N+2)!}&\cdots&-\left(-\frac{1}{2\lambda_{1I}}\right)^{N-j-1}
\frac{\vartheta_{1,0}^{j}}{j!}\\
0&-\binom{1}{0}\left(-\frac{1}{2\lambda_{1I}}\right)&\cdots&-\binom{j}{0}\left(-\frac{1}{2\lambda_{1I}}\right)^{j}
&-\binom{j+1}{0}\left(-\frac{1}{2\lambda_{1I}}\right)^{j+1}&\cdots&-\binom{N-1}{0}\left(-\frac{1}{2\lambda_{1I}}\right)^{N-1}\\
0&-\binom{1}{1}\left(-\frac{1}{2\lambda_{1I}}\right)&\cdots&-\binom{j}{1}\left(-\frac{1}{2\lambda_{1I}}\right)^{j}&
-\binom{j+1}{1}\left(-\frac{1}{2\lambda_{1I}}\right)^{j+1}&\cdots&-\binom{N-1}{1}\left(-\frac{1}{2\lambda_{1I}}\right)^{N-1}\\
\vdots&\vdots&\vdots&\vdots&\vdots&\vdots&\vdots\\
0&0&\cdots&-\binom{j}{j-1}\left(-\frac{1}{2\lambda_{1I}}\right)^{j}&-\binom{j+1}{j-1}\left(-\frac{1}{2\lambda_{1I}}\right)^{j+1}
&\cdots&-\binom{n-1}{j-1}\left(-\frac{1}{2\lambda_{1I}}\right)^{N-1}
\end{array}\right|\\
&\cdot\left|\begin{array}{cccc:cccc}1&0&\cdots&0&\binom{0}{0}\left(-\frac{1}{2\lambda_{1I}}\right)^0&0&\cdots&0\\
\vartheta_{1,0}^*&-\frac{1}{2\lambda_{1I}}&\cdots&0&\binom{1}{0}\left(-\frac{1}{2\lambda_{1I}}\right)^1&\binom{1}{1}\left(-\frac{1}{2\lambda_{1I}}\right)^1&\cdots&0\\
\vdots&\vdots&\vdots&\vdots&\vdots&\vdots&\vdots&\vdots\\
\frac{\vartheta_{1,0}^{*,j-1}}{(j-1)!}
&-\frac{1}{2\lambda_{1I}}\frac{\vartheta_{1,0}^{*,j-2}}{(j-2)!}&\cdots&\left(-\frac{1}{2\lambda_{1I}}\right)^{N-j-1}
\frac{\vartheta_{1,0}^{*,2j-n}}{(2j-N)!}&\binom{j-1}{0}\left(-\frac{1}{2\lambda_{1I}}\right)^{j-1}&\binom{j-1}{1}\left(-\frac{1}{2\lambda_{1I}}\right)^{j-1}
&\cdots&\binom{j-1}{j-1}\left(-\frac{1}{2\lambda_{1I}}\right)^{j-1}\\\cdashline{1-8}[4pt/5pt]
\frac{\vartheta_{1,0}^{*,j}}{j!}
&-\frac{1}{2\lambda_{1I}}\frac{\vartheta_{1,0}^{*,j-1}}{(j-1)!}&\cdots&\left(-\frac{1}{2\lambda_{1I}}\right)^{N-j-1}
\frac{\vartheta_{1,0}^{*,2j-N+1}}{(2j-N+1)!}&\binom{j}{0}\left(-\frac{1}{2\lambda_{1I}}\right)^{j}&\binom{j}{1}\left(-\frac{1}{2\lambda_{1I}}\right)^{j-1}
&\cdots&\binom{j}{j-1}\left(-\frac{1}{2\lambda_{1I}}\right)^{j}\\
\vdots&\vdots&\vdots&\vdots&\vdots&\vdots&\vdots&\vdots\\
\frac{\vartheta_{1,0}^{*,N-1}}{(N-1)!}&-\frac{1}{2\lambda_{1I}}\frac{\vartheta_{1,0}^{*,N-2}}{(N-2)!}
&\cdots&\left(-\frac{1}{2\lambda_{1I}}\right)^{N-j-1}\frac{\vartheta_{1,0}^{*,j}}{j!}
&\binom{N-1}{0}\left(-\frac{1}{2\lambda_{1I}}\right)^{N-1}&\binom{N-1}{1}\left(-\frac{1}{2\lambda_{1I}}\right)^{N-1}
&\cdots&\binom{N-1}{j-1}\left(-\frac{1}{2\lambda_{1I}}\right)^{N-1}\end{array}\right|\\&\cdot (-1)^{2j^2+3N^2+2N+1}\left(-\frac{\ii}{2\lambda_{1I}}\right)^{N-1}+\text{l.o.t}\\
\end{split}
\end{equation}
\normalsize
The leading order term in the second determinant in Eq.\eqref{eq:dendet} can be calculated with the same method in Eq.\eqref{eq:lead-order1}, while the first determinant has a little difference, whose leading order term can be given with a different choice, as shown with the dotted line. With the Lemma \ref{lemmaGamma}, Eq.\eqref{eq:dendet} becomes
\scriptsize
\begin{equation}
\begin{split}
&(-1)^{N}\left(-\frac{\ii}{2\lambda_{1I}}\right)^{N-1}\left|\begin{array}{cccc}
(-1)^{N-j+1}&0&\cdots&\left(-\frac{1}{2\lambda_{1I}}\right)^{N-j-2}\frac{\vartheta_{1,0}^{2j+2-N}}{(2j+2-N)!}\\
0&-\binom{1}{0}\frac{1}{2\lambda_{1I}}&\cdots&\binom{j}{0}\left(-\frac{1}{2\lambda_{1I}}\right)^{j}\\
0&-\binom{1}{1}\frac{1}{2\lambda_{1I}}&\cdots&\binom{j}{1}\left(-\frac{1}{2\lambda_{1I}}\right)^{j}\\
\vdots&\vdots&\vdots&\vdots\\
0&0&\cdots&\binom{j}{j-1}\left(-\frac{1}{2\lambda_{1I}}\right)^{j}
\end{array}\right|\left|\begin{array}{ccc}\frac{\vartheta_{1,0}^{j+1}}{(j+1)!}&\cdots&\frac{\vartheta_{1,0}^{N-1}}{(N-1)!}\\
-\frac{1}{2\lambda_{1I}}\frac{\vartheta_{1,0}^{j}}{j!}&\cdots&-\frac{1}{2\lambda_{1I}}\frac{\vartheta_{1,0}^{N-2}}{(N-2)!}\\
\vdots&\vdots&\vdots\\
\left(-\frac{1}{2\lambda_{1I}}\right)^{N-j-2}\frac{\vartheta_{1,0}^{2j-N+3}}{(2j-N+3)!}&\cdots&\left(-\frac{1}{2\lambda_{1I}}\right)^{N-j-2}
\frac{\vartheta_{1,0}^{j+1}}{(j+1)!}
\end{array}\right|\\&\cdot\left|\begin{array}{cccc}\binom{0}{0}\left(-\frac{1}{2\lambda_{1I}}\right)^{0}&0&\cdots&0\\
\binom{1}{0}\left(-\frac{1}{2\lambda_{1I}}\right)^{1}&\binom{1}{1}\left(-\frac{1}{2\lambda_{1I}}\right)^{1}&\cdots&0\\
\vdots&\vdots&\vdots&\vdots\\
\binom{j-1}{0}\left(-\frac{1}{2\lambda_{1I}}\right)^{j-1}&\binom{j-1}{1}\left(-\frac{1}{2\lambda_{1I}}\right)^{j-1}
&\cdots&\binom{j-1}{j-1}\left(-\frac{1}{2\lambda_{1I}}\right)^{j-1}
\end{array}\right|\left|
\begin{array}{cccc}
\frac{\vartheta_{1,0}^{*,j}}{j!}
&-\frac{1}{2\lambda_{1I}}\frac{\vartheta_{1,0}^{*,j-1}}{(j-1)!}&\cdots&\left(-\frac{1}{2\lambda_{1I}}\right)^{N-j-1}
\frac{\vartheta_{1,0}^{*,2j-n+1}}{(2j-n+1)!}\\
\vdots&\vdots&\vdots&\vdots\\
\frac{\vartheta_{1,0}^{*,N-1}}{(N-1)!}&-\frac{1}{2\lambda_{1I}}\frac{\vartheta_{1,0}^{*,N-2}}{(N-2)!}
&\cdots&\left(-\frac{1}{2\lambda_{1I}}\right)^{N-j-1}\frac{\vartheta_{1,0}^{*,j}}{j!}
\end{array}\right|+\text{l.o.t}\\&=(-1)^{j+1}(\vartheta_{1,0})^{(N-j-1)(j+1)}(\vartheta^*_{1,0})^{j(N-j)}
\left(-\frac{1}{2\lambda_{1I}}\right) ^{j^2+(N-j-1)^2}\left(
-\frac{\ii}{2\lambda_{1I}}\right) ^{N-1}\Gamma_{j}\Gamma_{j+1}+{\rm l.o.t}
\end{split}
\end{equation}
\normalsize
which is the leading order coefficient of $\ee^{2\theta_1^*}\ee^{2(N-j-1)(\theta_1+\theta_1^*)}$. It completes this proof.
\end{proof}
In the end, we give the proof for the lemma \ref{lemmamatch}:
\begin{proof}
Observing the proof for the lemma \ref{lemma-denom}, we know the leading order term of the denominator and the numerator contains polynomial function and exponential function. In order to find the asymptotic behavior when $t\to\pm\infty$, we should match the $t$ degree in both two functions. We only give the detailed calculation about $t\to+\infty$, conversely, $t\to-\infty$ can be given similarly. Now we choose any two terms in the denominator to show how to decide this characteristic curve:
\begin{equation}\label{eq:twolead}
\begin{split}
&\ee^{2(N-j_i)(\theta_1+\theta_1^*)}\mathscr{A}_{j_i}, (i=1,2).
\end{split}
\end{equation}
Suppose the characteristic curve is
\begin{equation}\label{eq:scale}
x=s+v_1t+\mu\log(t).
\end{equation}
If the high-order soliton moves along this characteristic curve, Eq.\eqref{eq:twolead} becomes
\begin{equation}\label{eq:two-terms}
\begin{split}
&|\mathscr{C}_{1}|^{2j_i(N-j_i)}\left((-\ii)^N\left(2\lambda_{1I}\right)^{-(N-j_i)^2-j_i^2}\right)
t^{2(N{-}j_i)(j_i{-}2\lambda_{1I}\mu)}\ee^{{-}4\lambda_{1I}(N{-}j_i)s}
\\+&O\left(t^{2(N{-}j_i)(j_i{-}2\lambda_{1I}\mu)-1}\log(t)\right),
\end{split}
\end{equation}

If both two terms in Eq. \eqref{eq:two-terms} can be matched well, they should have the same power of $t$, that is
$$2(N{-}j_1)(j_1{-}2\lambda_{1I}\mu)=2(N{-}j_2)(j_2{-}2\lambda_{1I}\mu),$$
which deduces
\begin{equation}\label{eq:mu}
\mu=\frac{j_1+j_2-N}{2\lambda_{1I}}.
\end{equation}

Moreover, we need to determine the constraint that $j_1$ and $j_2$ should satisfy. To achieve this aim, we discuss the leading order term of the denominator $\det(\mathbf{M})$ and numerator $\det(\mathbf{G})$ along this characteristic curve with given $\mu$ in Eq.\eqref{eq:mu}. The power of $t$ becomes
\begin{equation}\label{eq:numdenom}
\begin{split}
&t^{-2\left(j-\frac{j_1+j_2}{2}\right)^2+\frac{1}{2}\left(j_1+j_2-2N\right)^2}, \quad  j=0, 1,\cdots, N-1,\\
&t^{-2\left(j-\frac{j_1+j_2-1}{2}\right)^2-\frac{1}{2}\left(j_1+j_2-2N+1\right)\left(j_1+j_2-2N-1\right)},
\end{split}
\end{equation}
where the first formula in Eq.\eqref{eq:numdenom} is the denominator term and the second formula is the numerator term. It can be seen the power of $t$ is a quadratic function with respect to $j$. Obviously, the power of $t$ in the denominator can reach its maximum $\frac{(j_1+j_2-2N)^2}{2}$ at $j=\frac{j_1+j_2}{2}$, and the numerator can reach its maximum $\frac{(j_1+j_2-2N)^2-1}{2}$ at $j=\frac{j_1+j_2-1}{2}$. If $\frac{j_1+j_2}{2}$ is an integer, then the asymptotic behavior of $q^{[N]}$ along this curve Eq.\eqref{eq:scale} is vanishing as $t\to\infty$ because the power of $t$ in the denominator is higher than the one in the numerator. Otherwise, $\frac{j_1+j_2-1}{2}$ is an integer. In this case, if both terms with $j=j_1$ and $j=j_2$ become the leading order simultaneously, $j_1$ and $j_2$ must satisfy $|j_1-j_2|=1$. Under this condition, the asymptotic behavior of $q^{[N]}$ will tend to the single soliton, so we have $\mu=\frac{2j_1-1-N}{2\lambda_{1I}}, \left(j_1=1,2,\cdots, N\right)$, we rewrite $j_1=\kappa$. It completes this proof.
\end{proof}
\section*{Acknowledgement}
\lm{X.Z. acknowledges support from the China Postdoctoral Science Foundation(Grant No. 2020M682692),} L.L. is supported by the National Natural Science Foundation of China (Grant No. 11771151), the Guangzhou Science and Technology Program of China (Grant No. 201904010362), and the Fundamental Research Funds for the Central Universities of China (Grant No. 2019MS110),

\end{document}